\documentclass[11  pt, oneside]{article}   	
\usepackage[margin=1.25in,lmargin=0.75in,rmargin=1.75in,bmargin=0.75in,tmargin=0.75in]{geometry}      
\usepackage{graphicx}
\usepackage{booktabs} 

\usepackage[utf8]{inputenc}

\usepackage{xspace}
\usepackage{amsmath}
\providecommand{\subjectto}{\ensuremath{\text{subject to}}}

\DeclareMathOperator{\cut}{cut}
\DeclareMathOperator{\vol}{vol}
\DeclareMathOperator{\den}{density}

\newcommand{\alg}[1]{\textsc{#1}}
\newcommand{\fiveLP}{\alg{fiveLP}\xspace}
\newcommand{\lcc}{\alg{LambdaCC}\xspace}
\newcommand{\dorm}{\ensuremath{\mathit{dorm}}}
\newcommand{\yr}{\ensuremath{\mathit{year}}}
\newcommand{\fs}{\ensuremath{\mathit{s/f}}}

\newcommand{\lccb}{\lcc}
\graphicspath{{./}{./Figures/}}
\usepackage{subfig}

\usepackage{listings}
\usepackage{amsmath}
\usepackage{algorithm}
\usepackage[noend]{algpseudocode}
\usepackage{algorithmicx}
\usepackage{enumitem}
\usepackage{dgleich-math}
\usepackage{url}
\usepackage{microtype}

\newcommand{\threeLP}{\alg{threeLP} }
\newcommand{\twoCD}{\alg{twoCD} }

\newcommand{\nij}{\ensuremath{(i,j) \in E^-}}
\usepackage{color}
\newcommand{\tony}[1]{}
\renewcommand{\tony}[1]{{\textcolor{blue}{[{#1}\ --\ AIW]}}}

\definecolor{newcolor}{rgb}{.5, .0, .5}
\newcommand{\newln}[1]{}
\renewcommand{\newln}[1]{{{{#1}}}}

\usepackage[T1]{fontenc}
\usepackage[utf8]{inputenc}
\usepackage{authblk}

\begin{document}
\title{Unifying Sparsest Cut, Cluster Deletion, and Modularity Clustering Objectives with Correlation Clustering\footnote{A shorter version of this work was presented at the 2018 Web Conference~\cite{veldt2017lamcc}}}

\author[a]{Nate Veldt}
\author[b]{David F. Gleich} 
\author[c]{Anthony Wirth}

\affil[a]{Purdue University Mathematics Department}
\affil[b]{Purdue University Computer Science Department}
\affil[c]{The University of Melbourne, Computing and Information Systems School}

\maketitle

\begin{abstract}
Graph clustering, or community detection, is the task of identifying groups of closely related objects in a large network. In this paper we introduce
a new community-detection framework called \lcc that is based on a specially weighted version of correlation clustering. A key component in our
methodology is a clustering resolution parameter,~$\lambda$, which implicitly controls the size and structure of clusters formed by our framework. We
show that, by increasing this parameter, our objective effectively interpolates between two different strategies in graph clustering: finding a sparse cut and
forming dense subgraphs. Our methodology unifies and generalizes a number of other important clustering quality functions including modularity,
sparsest cut, and cluster deletion, and places them all within the context of an optimization problem that has been well studied from the perspective
of approximation algorithms. Our approach is particularly relevant in the regime of finding dense clusters, as it leads to a 2-approximation for the cluster deletion problem. We use our approach to cluster several graphs, including large collaboration networks and social networks.
\end{abstract}

\section{Introduction}
Identifying groups of related entities in a network is a ubiquitous task across scientific disciplines. This task is often called graph clustering, or community detection, and can be used to find similar proteins in a protein interaction network, group related organisms in a food web, identify communities in a social network, and classify web documents, among numerous other applications. 

Defining the right notion of a ``good'' community in a graph is an important precursor to developing successful algorithms for graph clustering. In
general, a good clustering is one in which nodes inside clusters are more densely connected to each other than to the rest of the graph. However, no
consensus exists as to the best way to determine the quality of network clusterings, and recent results show there cannot be such a consensus for the multiple possible reasons people may cluster data \cite{Peel2017ground}. Common objective functions studied by theoretical computer scientists include normalized cut, sparsest cut, conductance, and edge expansion, all of which measure some version of the cut-to-size ratio for a
single cluster in a graph. Other standards of clustering quality put a greater emphasis on the internal density of clusters, such as the cluster
deletion objective, which seeks to partition a graph into completely connected sets of nodes (cliques) by removing the fewest number of edges possible.

Arguably the most widely used multi-cluster objective for community detection is modularity, introduced by Newman and Girvan~\cite{newman2004modularity}. Modularity measures the difference between the true number of edges inside the clusters
of a given partitioning (``inner edges'') minus the \emph{expected} number of inner edges, where expectation is calculated with respect to a specific random graph model. 

There are a limited number of results which have begun to unify distinct clustering measures by introducing objective functions that are closely related to modularity and depend on a tunable clustering resolution parameter~\cite{delvenne2010stability,reichardt2006statistical}. Reichardt and Bornholdt developed an approach based on finding the minimum-energy state of an infinite range Potts spin glass. The resulting Hamiltonian function they study is viewed as a clustering objective with a resolution
parameter~$\gamma$, which can be used as a heuristic for detecting overlapping and hierarchical community structure in a network. When~$\gamma = 1$,
the authors prove an equivalence between minimizing the Hamiltonian and finding the maximum modularity partitioning of a
network~\cite{reichardt2006statistical}. Later, Delvenne et al.\ introduced a measure called the \emph{stability} of a clustering, which generalizes modularity and also is related to the normalized cut objective and Fiedler's spectral clustering method for certain values of an input parameter~\cite{delvenne2010stability}. 

The inherent difficulty of obtaining clusterings that are provably close to the optimal solution puts these objective functions at a disadvantage.
Although both the \emph{stability} and the Hamiltonian-Potts objectives provide useful interpretations for community detection,
there are no approximation guarantees for either: all current algorithms are heuristics.
Furthermore, it is known that maximizing modularity itself is not only NP-hard, but is also NP-hard to approximate to within any constant factor~\cite{dinh2016network}. 

\paragraph{Our Contributions}
In this paper, we introduce a new clustering framework based on a specially-weighted version of correlation
clustering~\cite{Bansal2004correlation}. Our partitioning objective for signed networks lies ``between'' the family of~$\pm 1$ complete instances
and the most general correlation clustering instances.
Our framework comes with several novel theoretical properties and leads to many connections between clustering
objectives that were previously not seen to be related. In summary, we provide:
\begin{itemize}[nosep]
	\item A novel framework \lcc for community detection that is related to modularity and the Hamiltonian, but is more amenable to approximation results.
	\item A proof that our framework interpolates between the sparsest cut objective and the cluster deletion problem,
	as we increase a single resolution parameter,~$\lambda$.
	\item Several successful algorithms for optimizing our new objective function in both theory and practice, including a 2-approximation for cluster deletion, which improves upon the previous best approximation factor of 3.
	\item A demonstration of our methods in a number of clustering applications, including social network analysis and mining cliques in collaboration networks.
\end{itemize}
\section{Background and Related Work}
Let~$G$ be an undirected and unweighted graph on~$n$ nodes~$V$, with~$m$ edges~$E$.
For all $v \in V$, let~$d_v$ be node~$v$'s degree. Given $S \subseteq V$, let
$\bar{S} = V\backslash S$ be the complement of~$S$ and $\vol(S) = \sum_{v\in
	S}d_v$ be its volume. For every two disjoint sets of vertices $S, T \subseteq V$,
$\cut(S,T)$ indicates the number of edges between~$S$ and~$T$. If $T =
\bar{S}$, we write $\cut(S) = \cut(S,\bar{S})$. Let~$E_S$ denote the interior
edge set of~$S$. The edge density of a cluster is $\den(S) =
|E_S|/ {|S| \choose 2 }$, the ratio between the number of edges to the
number of pairs of nodes in~$S$. By convention, the density of a single node is 1.
We now present background and related work that is foundational to our results, including definitions for several common clustering objectives.
\subsection{Correlation Clustering}
An instance of correlation clustering is given by a signed graph where every pair of nodes~$i$ and~$j$ possesses two non-negative weights,~$w_{ij}^+$ and~$w_{ij}^-$, to indicate how similar and how dissimilar~$i$ and~$j$ are, respectively. Typically only one of these weights is nonzero for each pair $i,j$. The objective can be expressed as an integer linear program (ILP):
\begin{equation}
\begin{array}{ll} \text{minimize}  & \sum_{i<j} w_{ij}^+ x_{ij} + w_{ij}^- (1-x_{ij})\\ \subjectto  & x_{ij} \leq x_{ik} + x_{jk}  \text{ for all $i,j,k$} \\ & x_{ij} \in \{0,1\} \text{ for all $i,j$.} \end{array}
\label{eq:cc}
\end{equation}
In the above formulation,~$x_{ij}$ represents ``distance'':
$x_{ij} = 0$ indicates that nodes~$i$ and~$j$ are clustered together, while $x_{ij} = 1$ indicates they are separated. Including triangle inequality constraints ensures the output of the above ILP defines a valid clustering of the nodes. This objective counts the total \emph{weight} of disagreements between the signed weights in the graph and a given clustering of its nodes.
The disagreement (or ``mistake'') weight of a pair~$i,j$ is~$w_{ij}^-$ if the nodes are clustered together, but~$w_{ij}^+$ if they are separated.
We can equivalently define the agreement weight to be $w_{ij}^+$ if $i,j$ are clustered together, but~$w_{ij}^-$ if they are separated.
The optimal clusterings for maximizing agreements and minimizing disagreements are identical, but it is more challenging to approximate the latter objective.

Correlation clustering was introduced by Bansal et al., who proved the problem is
NP-complete~\cite{Bansal2004correlation}. They gave a polynomial-time approximation
scheme for the maximization version and a constant-factor
approximation for minimizing disagreements in~$\pm1$-weighted graphs.
Subsequently, Charikar et al.\ gave a factor 4-approximation for minimizing disagreements and proved APX-hardness of this variant.
They also described an $O(\log n)$ approximation for minimization in general weighted graphs~\cite{charikar2005clustering},
proved independently by two different groups, who showed that minimizing disagreements is
equivalent to minimum multicut~\cite{Demaine2003,Emanuel2003}.

The problem has also been studied for the case where edges carry both positive and negative
weights, satisfying probability constraints: for all pairs~$i,j$,
$w_{ij}^+ + w_{ij}^- = 1$.
Ailon et al. gave a~$2.5$-approximation for this version of the problem based on an LP-relaxation, and
additionally developed a very fast algorithm, called \alg{Pivot}, that in expectation gives a
$3$-approximation~\cite{ailon2008aggregating}. Currently the best-known approximation factor for correlation clustering on~$\pm1$ instances
is slightly smaller than~$2.06$, obtained by a careful rounding of the canonical LP relaxation~\cite{chawla2015near}.


\subsection{Sparsest Cut and Normalized Cut}
One measure of cluster quality in an unsigned network~$G$ is the sparsest cut score, defined for a set~$S \subseteq V$ to be $\phi(S) = \cut(S)/|S| + \cut(S)/|\bar{S}| = n \cdot \cut(S)/ (|S||\bar{S}|)$.
Smaller values for~$\phi(S)$ are desirable, since they indicate that~$S$, in spite of its size, is only loosely connected to the rest of the graph.
This measure differs by at most a factor of two from the related edge expansion measure: $\cut(S)/(\min \{|S|, |\bar{S}|
\})$. If we replace $|S|$ with $\vol(S)$ in these two objectives, we obtain the normalized cut and the conductance
measure respectively. In our work we focus on a multiplicative scaling of the sparsest cut objective that we call the
\emph{scaled sparsest cut}: $\psi(S) = \phi(S)/n = {\cut(S)}/(|S| |\bar{S}|)$, which is identical to sparsest cut in terms of multiplicative approximations. The best known approximation for finding the minimum sparsest cut of a graph is an $O(\sqrt{\log n})$-approximation algorithm due to Arora et al.~\cite{arora2009expander}.

\subsection{Modularity and the Hamiltonian}
One very popular measure of clustering quality is modularity, introduced in its most basic form by Newman and Girvan~\cite{newman2004modularity}. We more closely follow the presentation of modularity given by Newman~\cite{newman2006finding}.
The modularity~$Q$ of an underlying clustering is:
\begin{equation}
\label{modularity}
{Q(x) = \frac{1}{2m} \sum_{i \neq j} \left( A_{ij} - P_{ij} \right) (1-x_{ij})\,},
\end{equation}
where $A_{ij} = 1$ if nodes~$i$ and~$j$ are adjacent, and zero otherwise, and~$x_{ij}$ is again the binary variable indicating ``distance'' between~$i$ and~$j$ in the corresponding clustering. The value~$P_{ij}$ represents the probability of an edge existing between~$i$ and~$j$ in a specific random graph model.
The intent of this measure is to reward clusterings in which the actual number of edges inside a cluster is greater than the \emph{expected} number of edges in the cluster, as determined by the choice for~$P_{ij}$. Although there are many options, it is standard in the literature to set $P_{ij} ={d_id_j}/(2m)$, since this preserves both the degree distribution and the expected number of edges between the original graph and null model.
Many generalizations have been introduced for modularity, including an extension to multislice networks, which allow one to study the evolution of communities in a network over time~\cite{mucha2010community}.

By slightly editing the modularity function, we obtain the Hamiltonian objective of Reichardt and Bornholdt~\cite{reichardt2006statistical}:
\begin{equation}
\label{Hamiltonian}
{ \mathcal{H}(x) = - \sum_{i\neq j} \left(A_{ij} - \gamma P_{ij} \right)(1-x_{ij})\,}.
\end{equation}
The primary difference between this and modularity is the inclusion of a clustering resolution parameter~$\gamma$. If we fix $\gamma = 1$, minimizing~(\ref{Hamiltonian}) is equivalent to maximizing modularity. When varied, this parameter controls how much a clustering is penalized for putting two non-adjacent nodes together or separating adjacent nodes. 
{Recently Jeub et al. presented a new strategy for sampling values of this resolution parameter to produce very good hierarchical clusterings of an input graph without resorting to ad-hoc methods for finding appropriate values for $\gamma$~\cite{jeub2017multiresolution}.}



The Hamiltonian objective is in turn closely related to the \emph{stability} of a clustering as defined by Delvenne et al., another generalization of modularity~\cite{delvenne2010stability}. Roughly speaking, the stability of a partition measures the likelihood that a random walker, beginning at a node and following outgoing edges uniformly at random, will end up in the cluster it started in after a random walk of length~$t$. This~$t$ serves as a resolution parameter, since the walker will tend to ``wander" farther when~$t$ is increased, leading to the formation of larger clusters when the stability is maximized. Delvenne et al.\ showed that objective~(\ref{Hamiltonian}) is equivalent to a linearized version of the stability measure for a specific range of time steps~$t$~\cite{delvenne2010stability}.

\newln{A number of equivalence results between modularity and other clustering objectives have been noted in previous work. Agarwal et al.\ showed that modularity is equivalent at optimum to a special case of correlation clustering~\cite{agarwal2008metricmod}. Newman demonstrated that maximizing modularity with a resolution parameter is equivalent to maximizing a log-likelihood function for the degree-corrected stochastic block model~\cite{newman2013equivalence}. Finally, an equivalence between a normalized version of modularity and a multi-cluster generalization of normalized cut has been independently shown by a number of authors~\cite{bolla2011penalized,wangModularityncut1,yuModularityncut2}.}

\subsection{Cluster Deletion} 
Cluster deletion is the problem of finding a minimum number of edges in~$G$ to be deleted in order to convert~$G$ into a disjoint set of cliques. This
can be viewed as stricter version of correlation clustering, in which we want to minimize disagreements, but we are strictly prohibited from making mistakes at negative edges. This problem was first studied by Ben-Dor et al.~\cite{bendor1999clustering}, later
formalized in the work of Natanzon et al.~\cite{natanzon1999complexity}, who proved it is NP-hard,
and Shamir et al.~\cite{shamir2004cluster}, who showed it is APX-hard.
The latter studied the problem in conjunction with
other related \emph{edge-modification} problems, including cluster completion
and cluster editing.

Numerous fixed parameter tractability results are known for cluster deletion
\cite{bocker2011evenfaster,Gramm2003,Gramm2004,Damaschke2009}, as well many
results regarding special graphs for which the problem can be solved in
polynomial time \cite{gao2013cluster,bonomo2015complexity,
dessmark2007edgeclique,bonomo2015one}. 
\newln{Dessmark et al.\ gave an $O(\log n)$ approximation for the problem when the edges have arbitrary weights, and proved that in the unweighted case, recursively finding maximum cliques will return a clustering with a cluster deletion score within a factor 2 of optimal~\cite{dessmark2007edgeclique}. In general however this latter procedure is NP-hard. Charikar et al.\ showed that a slight adaptation of their correlation clustering algorithm produces a 4-approximation~\cite{charikar2005clustering}. We note finally that although van Zuylen and Williamson make no explicit mention of cluster deletion, their results for constrained correlation clustering imply a 3-approximation for the problem (see Theorem 4.2 in~\cite{zuylen2009deterministic}).}
\section{Theoretical Results}
%
Our novel clustering framework takes an unsigned graph $G=(V,E)$ and converts it into a signed graph~$G'=(V,E^+,E^-)$ on the same set of nodes,~$V$,
for a fixed clustering resolution parameter $\lambda \in (0,1)$. Partitioning~$G'$ with respect to the correlation clustering objective will then induce a clustering on~$G$. To construct the signed graph, we first introduce a node weight~$w_v$ for each $v \in V$. If $(i,j)\in E$, we place a positive edge between nodes~$i$ and~$j$ in~$G'$, with weight $(1-\lambda w_iw_j)$. For $(i,j) \notin E$, we place a negative edge between~$i$ and~$j$ in~$G'$, with weight $\lambda w_iw_j$. We consider two different choices for node weights~$w_v$: setting $w_v = 1$ for all $v$ (\emph{standard}) or choosing $w_v = d_v$ (\emph{degree-weighted}). 
In Figure~\ref{lccpic} we illustrate the process of converting~$G$ into the \lcc signed graph,~$G'$. The goal of \lcc is to find the clustering that minimizes disagreements in $G'$, or equivalently minimizes the following objective function expressed in terms of edges and non-edges in $G$:
\begin{equation}
\label{lccgen}
\lambda \mathcal{CC}(x) =  \sum_{ \substack{(i,j) \in E}} (1 - \lambda w_i w_j) x_{ij} + \sum_{ \substack{(i,j) \notin E}} \lambda w_i w_j (1-x_{ij})
\end{equation}
$x = (x_{ij})$ represents the binary distances for the clustering.
\begin{figure}
	\centering
	\includegraphics[width=\linewidth]{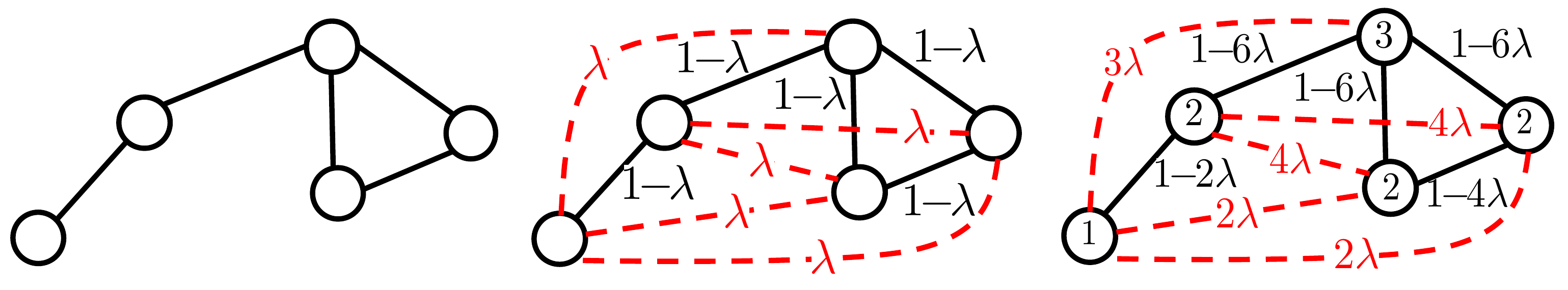}
	\caption{We convert a toy graph (left) into a signed graph for standard (middle) and degree-weighted (right) \lcc. Dashed red lines indicate negative edges. Partitioning the signed graph via correlation clustering induces a clustering on the original unsigned graph.}
	\label{lccpic}
\end{figure}
\subsection{Connection to Modularity}
Despite a significant difference in approach and interpretation, the clustering that minimizes disagreements is the same clustering that minimizes the
Hamiltonian objective~(\ref{Hamiltonian}), for a certain choice of parameters. To see this, we introduce node adjacency variables $A_{ij}$ in objective~\eqref{lccgen} and perform a few steps of algebra:
\begin{align*}
\lambda \mathcal{CC}(x) &=  \sum_{(i,j) \in E} (A_{ij} - \lambda w_i w_j)x_{ij} - \sum_{(i,j) \notin E} (A_{ij} - \lambda w_i w_j )(1-x_{ij}) \\
&= \sum_{(i,j) \in E} (A_{ij} - \lambda w_i w_j)x_{ij} - \sum_{(i,j) \in E} (A_{ij} - \lambda w_i w_j)\\
& \hspace{.5cm} + \sum_{(i,j) \in E} (1 - \lambda w_i w_j) - \sum_{(i,j) \notin E} (A_{ij} - \lambda w_i w_j)(1-x_{ij}) \\
&= \sum_{(i,j) \in E} (1 - \lambda w_i w_j) - \sum_{i<j} (A_{ij} - \lambda w_i w_j)(1-x_{ij})\,.
\end{align*}

Choosing $P_{ij} = {w_iw_j}/(2m)$ and $\gamma = 2m\lambda$, we see that:
\begin{equation}
\label{equivSI}
\lambda \mathcal{CC}(x) = \sum_{(i,j) \in E} (1 - \lambda w_i w_j) + \frac{\mathcal{H}(x)}{2}\,,
\end{equation}
where the first term is just a constant. This theorem follows:

\begin{theorem}
Minimizing disagreements for the \lcc objective is equivalent to minimizing~$\mathcal{H}(x)$.
\end{theorem}
The choice~$P_{ij} = {w_iw_j}/(2m)$ is reminiscent of the graph null model most commonly used for modularity and the Hamiltonian. This best highlights the similarity between these objectives and degree-weighted \lccb. 

\subsection{Standard \lcc}
While degree-weighted \lcc is more closely related to modularity and the Hamiltonian, standard \lcc (setting $w_v = 1$ for every $v \in V$) 
leads to strong connections between the sparsest cut objective and cluster deletion. This version corresponds to solving a correlation clustering problem where all positive edges have equal weight,~$(1-\lambda)$, while all negative edges have equal weight,~$\lambda$. The objective function for minimizing disagreements is
\begin{equation}
\label{lamobjstand}
\min \sum_{(i,j)\in E^+} (1-\lambda) x_{ij} + \sum_{(i,j) \in E^-} \lambda (1-x_{ij})\,,
\end{equation}
where we include the same constraints as in ILP~(\ref{eq:cc}). This is a strict
generalization of the unit-weight correlation clustering problem~\cite{Bansal2004correlation} ($\lambda = 1/2$)
indicating the problem in general is NP-hard (though it admits several approximation algorithms).
If~$\lambda$ is~$0$ or~$1$, the problem is trivial to solve: put all
nodes in one cluster or put each node in a singleton cluster, respectively. By selecting values
for~$\lambda$ other than~$0$, $1/2$, or~$1$, we uncover
subtler connections between identifying sparse cuts and finding dense subgraphs in the network.
\subsection{Connection to Sparsest Cut}
Given~$G$ and~$\lambda$, the weight of positive-edge mistakes in the \lcc objective made by a two-clustering $\mathcal{C} = \{S, \bar{S}\}$ equals the weight of edges crossing the cut: $(1-\lambda)\cut(S)$. To compute the weight of negative-edge mistakes, we take the weight of all negative edges in the entire network, $\lambda \left({n \choose 2} - |E|  \right)$, and then subtract the weight of negative edges between $S$ and $\bar{S}$: $\lambda\left(|S| |\bar{S}| - \cut(S) \right)$. Adding together all terms we find that the \lcc objective for this clustering is
\begin{equation}
\label{eq:2cutobj}
 { \cut(S,\bar{S}) - \lambda |S| |\bar{S}| +  \lambda {n
	\choose 2} -\lambda |E|\,}.
\end{equation}
Note that if we minimize~\eqref{eq:2cutobj} over all 2-clusterings, we solve the decision version of the minimum scaled sparsest cut problem: a few steps of algebra confirm that there is some set $S \subseteq V$ with $\psi(S) = \cut(S)/(|S| |\bar{S}|) < \lambda$ if and only if (\ref{eq:2cutobj}) is less than $\lambda {n \choose 2} -\lambda |E|$.

In a similar way we can show that objective~\eqref{lamobjstand} is equivalent to
\begin{equation}
\label{eq:cutobj}
\min \,\, \frac{1}{2}\sum_{i=1}^k \cut(S_i) - \frac{\lambda}{2} \sum_{i=1}^k
|S_i| |\bar{S_i}| +  \lambda {n \choose 2} -\lambda |E|\,,
\end{equation}
where we minimize over all clusterings of $G$ (note that the number of clusters $k$ is determined automatically by optimizing the objective).
 In this case, optimally solving objective~(\ref{eq:cutobj}) will tell us whether we can find a clustering $\mathcal{C} = \{S_1, S_2, \hdots , S_k\}$ such that
\[\frac{\sum_{i=1}^k \cut(S_i,\bar{S_i})}{ \sum_{j=1}^k |S_j| |\bar{S_j}|} < \lambda.\]
Hence \lcc can be viewed as a multi-cluster generalization of the decision version of minimum sparsest cut. We now prove an even deeper connection between sparsest cut and \lcc. Using degree-weighted \lcc yields an analogous result for normalized cut.
\begin{theorem}
	\label{thm:lambound}
	Let~$\lambda^*$ be the minimum scaled sparsest cut for a graph $G$.
	\begin{enumerate}[label=(\alph*)]
		\item For all $\lambda > \lambda^*$, optimal solution~(\ref{eq:cutobj}) partitions~$G$ into two or more clusters,
		each of which has scaled sparsest cut~$\leq \lambda$. There exists some~$\lambda' > \lambda^*$ such that the optimal
clustering for \lcc is the minimum sparsest cut partition.
		\item For~$\lambda \leq \lambda^*$, it is optimal to place all nodes into a single cluster.
	\end{enumerate}

\end{theorem}
\begin{proof}
\textbf{Statement (a)} Let~$S^*$ be some optimal sparsest cut-inducing set in~$G$, i.e., \[\psi(S^*) = {\cut(S^*)}/ ({|S^*| | \bar{S}^*|}) = \lambda^*.\]
The \lcc objective corresponding to~$\mathcal{C} = \{S^*, \bar{S}^* \}$ is
\begin{equation}
\label{optssc}
\cut(S^*) - \lambda|S^*| | \bar{S}^*| + \lambda {n \choose 2} -\lambda|E|\,.
\end{equation} 
When minimizing objective~(\ref{eq:cutobj}), we can always obtain a score of $\lambda {n \choose 2} -\lambda |E|$
by placing all nodes into a single cluster. Note however that the score of clustering $\{S^*, \bar{S}^*\}$ in expression~\eqref{optssc} is strictly less than $\lambda {n \choose 2} -\lambda|E|$ for all $\lambda > \lambda^*$. Even if $\{S^*, \bar{S}^*\}$ is not optimal, this means that when $\lambda>\lambda^*$, we can do strictly better than placing all nodes into one cluster. In this case let $\mathcal{C}^*$ be the optimal \lcc clustering and consider two of its clusters:~$S_i$ and~$S_j$. The weight of disagreements between~$S_i$ and~$S_j$ is equal to the number of positive edges between them times the weight of a positive edge: $(1-\lambda)\cut(S_i,S_j)$. Should we form a new clustering by merging~$S_i$ and~$S_j$, these
positive disagreements will disappear; in turn, we would introduce
$\lambda |S_i||S_j| - \lambda \cut(S_i,S_j)$ new mistakes, being negative edges between the clusters.
Because we assumed $\mathcal{C}^*$ is optimal, we know that we cannot decrease the objective by merging two of the clusters, implying that
\begin{equation*}
	(1-\lambda)\cut(S_i,S_j) - \left(\lambda |S_i||S_j| - \lambda \cut(S_i,S_j)\right) = \cut(S_i,S_j) - \lambda|S_i||S_j| \leq 0\,. 
\end{equation*}
Given this, we fix an arbitrary cluster~$S_i$ and perform a sum over all other clusters to see that
\[ { \sum_{j \neq i}} \cut(S_i,S_j) - {\textstyle \sum_{j \neq i}} \lambda |S_i||S_j| \leq 0\, \]
\[\implies \cut(S_i, \bar{S_i}) - \lambda |S_i||\bar{S_i}| \leq 0 \implies {\cut(S_i, \bar{S_i})}/\left({|S_i||\bar{S_i}|}\right)
\leq \lambda\,, \]
proving the desired upper bound on scaled sparsest cut.

Since~$G$ is a finite graph, there are a finite number of scaled sparsest cut
scores that can be induced by a subset of~$V$. Let~$\tilde{\lambda}$ be the second-smallest scaled sparsest cut score achieved,
so $\tilde{\lambda} > \lambda^*$. If we set $\lambda' = (\lambda^* + \tilde{\lambda})/2$, then the optimal \lcc clustering
produces at least two clusters, since $\lambda' > \lambda^*$, and each cluster has scaled sparsest cut at most $\lambda' < \tilde{\lambda}$.
By our selection of~$\tilde{\lambda}$, all clusters returned must have scaled sparsest cut exactly equal to~$\lambda^*$,
which is only possible if the clustering returned has two clusters. Hence this clustering is a minimum sparsest cut partition of the network.

\textbf{Statement (b)}
If $\lambda < \lambda^*$, forming a single cluster must be optimal, otherwise we could invoke Statement~(a) to assert the existence of some nontrivial cluster with scaled sparsest cut less than or equal to $\lambda < \lambda^*$, contradicting the minimality of~$\lambda^*$. If $\lambda = \lambda^*$, forming a single cluster or using the clustering $\mathcal{C} = \{S^*, \bar{S}^* \}$ yield the same objective score, which is again optimal for the same reason.
\end{proof}

\subsection{Connection to Cluster Deletion}
For large~$\lambda$ our problem becomes more similar to cluster deletion. We can reduce any cluster deletion problem to correlation
clustering by taking the input graph~$G$ and introducing a negative edge of weight~``$\infty$''
between every pair of non-adjacent nodes. This guarantees that optimally solving correlation clustering will yield clusters that all correspond to cliques in~$G$. Furthermore, the weight of disagreements will be the number of edges in~$G$ that are cut, i.e., the cluster deletion score. We can obtain a generalization of cluster deletion by instead choosing the weight of each negative edge to be $\alpha < \infty$. The corresponding objective is
\begin{equation}
\label{alpobj}
\sum_{(i,j)\in E^+} x_{ij} + \sum_{(i,j) \in E^-} \alpha (1-x_{ij})\,.
\end{equation}
If we substitute $\alpha = \lambda/ (1-\lambda)$ we see this differs from objective~(\ref{lamobjstand}) only by a multiplicative constant, and is
therefore equivalent in terms of approximation. When $\alpha > 1$, putting dissimilar nodes together will be more expensive than cutting positive
edges, so we would expect that the clustering which optimizes the \lcc objective will separate~$G$ into dense clusters that are ``nearly'' cliques. We formalize this with a simple theorem and corollary.
\begin{theorem}
	\label{thm:dense}
	If~$\mathcal{C}$ minimizes the \lcc objective for the unsigned network $G = (V,E)$, then the edge density of every cluster in~$\mathcal{C}$ is at least~$\lambda$.
\end{theorem}
\begin{proof}
	Take a cluster~$S \in \mathcal{C}$ and consider what would happen if we broke apart~$S$ so that each of its nodes
	were instead placed into its own singleton cluster.
	This means we are now making mistakes at every positive edge previously in $S$, which increases the weight of disagreements by~$(1-\lambda)|E_S|$. On the other hand, there are no longer negative mistakes between
	nodes in~$S$, so the \lcc objective would simultaneously decrease by
	$\lambda \left({|S| \choose 2} - |E_S| \right)$.
	The total change in the objective made by pulverizing~$S$ is
	\[{\textstyle (1-\lambda)|E_S| -  \lambda \left({|S| \choose 2} - |E_S| \right) = |E_S| - \lambda {|S| \choose 2}\, }, \]
	which must be nonnegative, since~$\mathcal{C}$ is optimal, so~$|E_S| - \lambda {|S| \choose 2}  \geq 0 \implies \text{density}(S) ={|E_S|}/{ {|S| \choose 2} } \geq \lambda$.
\end{proof}
%
%
%
\begin{corollary}
	Let~$G$ have~$m$ edges. For every $\lambda > m/(m+1)$, optimizing \lcc is equivalent to optimizing cluster deletion. 
\end{corollary}
\begin{proof}
All output clusters must have density at least $m/(m+1)$, which is only possible if the density is actually one,
since~$m$ is the total number of edges in the graph. Therefore all clusters are cliques and the \lcc and cluster deletion objectives differ only by a multiplicative constant $(1-\lambda)$.
\end{proof}

\subsection{Equivalences and Approximations}
We summarize the equivalence relationships between \lcc and other objectives in Figure~\ref{fig:lcc}. \newln{We additionally note that the results of Newman~\cite{newman2013equivalence} imply that \lcc is also equivalent to the log-likelihood function for the stochastic block model. This holds both in the case of degree-corrected SBM (which is equivalent to degree-weighted \lcc) and the standard SBM (corresponding to standard \lcc)}.
Accompanying Figure~\ref{fig:lcc}, Table~\ref{tab:approx} outlines the best-known approximation results both for maximizing agreements and minimizing disagreements
for the standard \lcc signed graph. For degree-weighted \lcc, the best-known approximation factors for all~$\lambda$ are $O(\log n)$~\cite{Demaine2003,Emanuel2003,charikar2005clustering}
for minimizing disagreements, and~$0.7666$ for maximizing agreements~\cite{swamy2004correlation}. Thus, \lcc is more amenable to approximation than modularity (and relatives) because of additive constants.

\begin{figure}
	\centering
	\includegraphics[width=.65\linewidth]{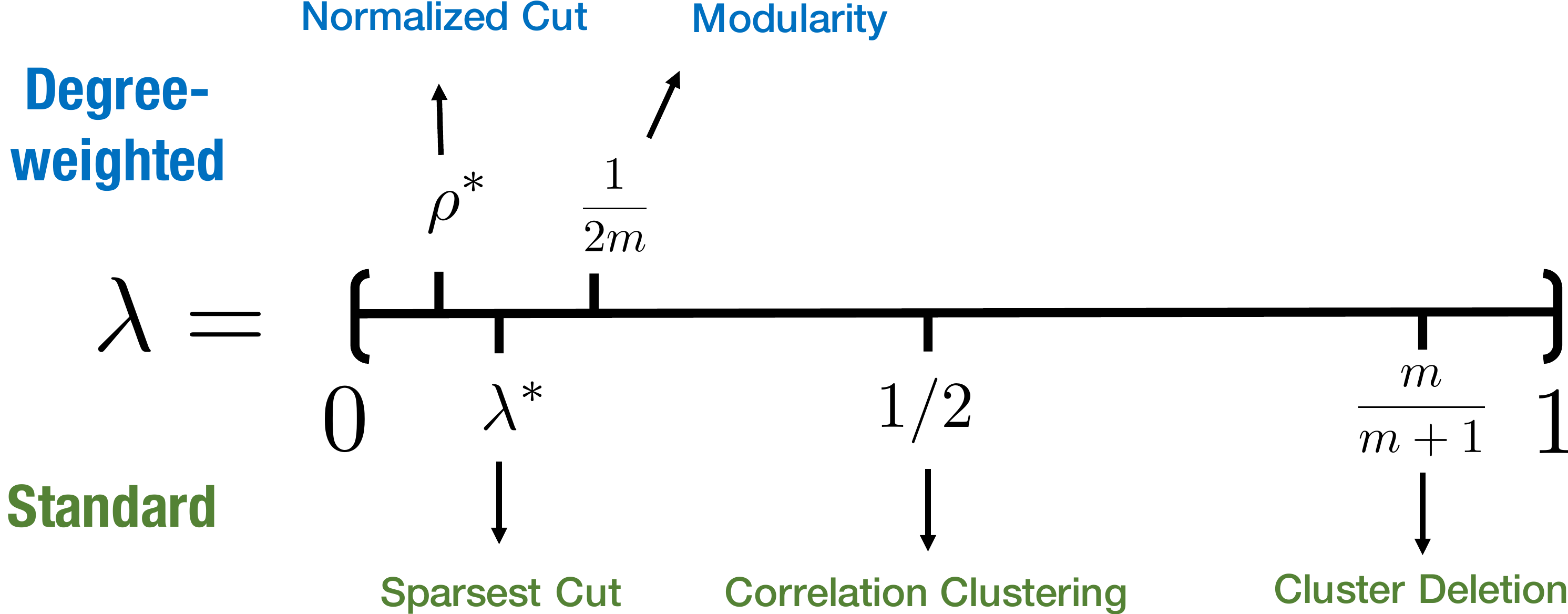}
	\caption{\lcc is equivalent to several other objectives for specific values of $\lambda \in (0,1)$.
	Values~$\lambda^*$ and~$\rho^*$ are not known a priori, but can be obtained by solving \lcc for increasingly smaller values of $\lambda$.}
	\label{fig:lcc}
\end{figure}
\begin{table}
	\centering
	\caption{The best approximation factors known for standard \lcc, for $\lambda \in (0,1)$, both for minimizing disagreements and
maximizing agreements. We contribute two constant-factor approximations for minimizing disagreements when~$\lambda > 1/2$.}
	\begin{tabular}{llll}
		\toprule
		  & $\lambda \in \left(0,{1}/{2}\right)$ & $\lambda = {1}/{2}$ & $\lambda \in \left({1}/{2},1\right)$ \\
		\midrule
		Max-Agree& $0.7666$~\cite{swamy2004correlation}&PTAS~\cite{Bansal2004correlation} & $0.7666$~\cite{swamy2004correlation}\\
		\midrule
		Min-Dis. & $O(\log n)$~\cite{charikar2005clustering,Demaine2003,Emanuel2003} & $2.06$~\cite{chawla2015near} &
$3 \left(2: \lambda > \frac{m}{m+1} \right)$ \\
		\bottomrule
	\end{tabular}
	\label{tab:approx}
\end{table}
\section{Algorithms}
We present several new algorithms for our \lcc framework, beginning with methods based on linear programming relaxations. Our best results are a 3-approximation for \lcc when $\lambda > 1/2$ and a 2-approximation for cluster deletion, which rely on a key theorem of van Zuylen and Willamson~\cite{zuylen2009deterministic}. We also show how to alter the approach of Charikar et al.~\cite{charikar2005clustering} for unweighted correlation clustering in order to obtain a 5-approximation for \lcc when $\lambda > 1/2$ and a related 4-approximation for cluster deletion, with self-contained proofs for these approximation guarantees. Although applying the techniques of van Zuylen and Williamson or Charikar et al.\ lead to different approximation guarantees, it is interesting to note that both approaches yield constant-factor approximations for $\lambda > 1/2$, but fail to yield approximation guarantees for arbitrarily small $\lambda$. This presents an interesting open question of whether a similar approximation guarantee can be obtained for small $\lambda$, or whether the computational complexity of the problem is fundamentally different. 


After presenting our LP-based methods, we also outline several more scalable heuristic techniques for solving our objective in practice.

\subsection{Randomized and Deterministic Pivoting Algorithms}
\begin{algorithm}[tb]
	\caption{\alg{CC-Pivot}}
	\begin{algorithmic}[5]
		\State{\bfseries Input:} Signed graph $G = (V,E^+,E^-)$
		\State {\bfseries Output:} Clustering~$\mathcal{C} = \alg{CC-Pivot}(G)$
		\State Select a pivot node $k \in V$
		\State Form cluster $S = \{v \in V: (k,v) \in E^+ \}$
		\State Output clustering $\mathcal{C} = \{S, \alg{CC-Pivot}(G - S) \}$
	\end{algorithmic}
	\label{alg:piv}
\end{algorithm}
Before presenting our methods we review a simple correlation clustering algorithm and a key theorem of van Zuylen and Williamson~\cite{zuylen2009deterministic}. Algorithm \ref{alg:piv} outlines the \alg{CC-Pivot} method, which selects an unclustered node from a signed graph $G$, clusters it with all of its positive neighbors, and recurses on the remaining unclustered nodes. We use notation $\alg{CC-Pivot}(G)$ to indicate applying this procedure on the graph $G$. If the pivot node is chosen uniformly at random, this corresponds to the 3-approximation for unweighted correlation clustering presented by Ailon et al.~\cite{ailon2008aggregating}. The following theorem shows how to obtain deterministic pivoting algorithms by first solving the LP-relaxation of correlation clustering and selecting pivot nodes more carefully at each step:

\begin{theorem}\label{thm3pt1} {(Theorem 3.1 in \cite{zuylen2009deterministic})}
	Let $G = (V,W^+, W^-)$ be a signed, weighted graph where each pair of nodes $(i,j)$ has positive and negative weights $w_{ij}^+ \in W^+$ and $w_{ij}^- \in W^-$. Given a set of budgets $\{c_{ij} : i \in V, j \in V, i\neq j \}$, and an unweighted graph $\tilde{G} = (V, F^+, F^-)$
	satisfying the following assumptions:
	\begin{enumerate}[label=(\roman*)]
		\item $w_{ij}^- \leq \alpha c_{ij}$ for all $(i,j) \in F^+$ and 
		
		$w_{ij}^+ \leq \alpha c_{ij}$ for all $(i,j) \in F^-$,
		\item $w_{ij}^+ + w_{jk}^+ + w_{ik}^- \leq \alpha \left(   c_{ij} + c_{jk} + c_{ik} \right)$ 
		
		for every triplet $\{i,j,k\}$ in $\tilde{G}$ with $(i,j), (j,k) \in F^+$, $ (i,k)\in F^- $,
	\end{enumerate}
	then applying \alg{CC-Pivot} on $\tilde{G}$ will return a solution that costs  at most $\alpha \sum_{i<j} c_{ij}$ if we choose a pivot $k$ that minimizes:
	\[ \frac{\sum_{(i,j)\in T_k^+(G)} w_{ij}^+ + \sum_{(i,j) \in T_k^-(G)} w_{ij}^-}{\sum_{(i,j) \in T_k^+(G) \cup T_k^-(G)} c_{ij}}. \]
	where
	\[ T_k^+(G) = \{ (i,j) \in F^+  : (k,j) \in F^-, (k,i) \in F^+  \} \]
	\[ T_k^-(G) = \{ (i,j) \in F^-  : (k,j) \in F^+, (k,i) \in F^+  \}. \]
\end{theorem}
The budgets $c_{ij}$ in the above theorem correspond to the cost $c_{ij} = w_{ij}^+x_{ij} + w_{ij}^-(1-x_{ij})$ of a pair $i,j$ in the LP-relaxation for correlation clustering:
\begin{equation}
\begin{array}{ll} \text{minimize}  & \sum_{i<j} w_{ij}^+ x_{ij} + w_{ij}^- (1-x_{ij})\\ \subjectto  & x_{ij} \leq x_{ik} + x_{jk}  \text{ for all $i,j,k$} \\ & 0 \leq x_{ij} \leq 1 \text{ for all $i,j$.} \end{array}
\label{eq:cclp}
\end{equation}
A full proof of the theorem is included in the original work~\cite{zuylen2009deterministic}. As the authors note, the same approximation result holds in expectation if pivot nodes are chosen uniformly at random.


\subsection{3-Approximation for \lcc}
\begin{algorithm}[tb]
	\caption{\threeLP}
	\begin{algorithmic}[5]
		\State{\bfseries Input:} Signed graph $G' = (V,E^+,E^-)$, $\lambda  \in (0,1)$
		\State {\bfseries Output:} Clustering~$\mathcal{C}$ of~$G'$
		\State Solve the LP-relaxation of ILP~\eqref{lamobjstand}, obtaining \emph{distances}~$(x_{ij})$
		\State Define $\tilde{G} = (V,F^+,F^-)$ where
		\[ F^+ = \{(i,j) : x_{ij} < 1/3 \}, \hspace{.5cm} F^- = \{(i,j) : x_{ij} \geq 1/3 \}  \]
		\State Return \alg{CC-Pivot}($\tilde{G}$).
	\end{algorithmic}
	\label{alg:LP3}
\end{algorithm}
We slightly alter the approach of van Zuylen and Williamson for unweighted correlation clustering~\cite{zuylen2009deterministic} to obtain an approximation algorithm for \lcc when $\lambda > 1/2$. Pseudocode for this method is displayed in Algorithm~\ref{alg:LP3}, which we call \threeLP since it satisfies the following guarantee:
\begin{theorem}
	Algorithm \threeLP satisfies Theorem~\ref{thm3pt1} with $\alpha = 3$ for standard \lcc when $\lambda > 1/2$.
\end{theorem}
\begin{proof}
	The first two inequalities we need to check for Theorem~\ref{thm3pt1} are
	\begin{equation}
	\label{one}
	\text{$w_{ij}^- \leq \alpha c_{ij}$ for all $(i,j) \in F^+$}
	\end{equation}
	\begin{equation}
	\label{two}
	\text{$w_{ij}^+ \leq \alpha c_{ij}$ for all $(i,j) \in F^-$}
	\end{equation}
	We must first understand what these terms mean for our specific problem. Our input graph $G' = (V,E^+,E^-)$ is made up simply of positive edges of weight $(1-\lambda)$ and negative edges with weight $\lambda$. If $(i,j) \in E^+$, the node pair has weights $(w_{ij}^+, w_{ij}^-) = (1-\lambda, 0)$, and LP cost $c_{ij} = (1-\lambda)x_{ij}$ in~\eqref{lamobjstand}. For negative edges $(i,j) \in E^-$, we have $(w_{ij}^+, w_{ij}^-) = (0,\lambda)$, and $c_{ij} = \lambda(1-x_{ij})$. By construction, if $(i,j) \in F^+$, then $x_{ij} < 1/3$, otherwise $(i,j) \in F^-$ and we know $x_{ij} \geq 2/3$. 
	
	Consider an edge $(i,j) \in F^+\cap E^+$. Then $w_{ij}^- = 0$ and inequality~\eqref{one} is trivial since the left hand side is zero. Similarly, inequality~\eqref{two} is trivial if $(i,j) \in F^-\cap E^-$. Assume then that $(i,j) \in F^+\cap E^-$. Then $w_{ij}^- = \lambda$ and $c_{ij} = \lambda (1 - x_{ij})$, and we know $x_{ij} < 1/3 \implies (1-x_{ij}) > 2/3$. Therefore:
	\[ w_{ij}^- = \lambda  < 3\lambda \left(2/3\right) < 3\lambda(1-x_{ij})= \alpha c_{ij}. \]
	On the other hand, if $(i,j) \in F^- \cap E^+$, then $w_{ij}^+ = (1-\lambda)$, $c_{ij} = (1-\lambda) x_{ij}$, and $x_{ij} \geq 1/3$, so we see:
	\[w_{ij}^+ = (1-\lambda) = 3(1-\lambda)\left(1/3 \right) \leq 3 (1-\lambda) x_{ij} = \alpha c_{ij}. \]
	This concludes the proof for inequalities~\eqref{one} and \eqref{two}. Next we consider a triplet of nodes $\{i,j,k\}$ where $(i,j) \in F^+, (j,k) \in F^+$ but $(i,k) \in F^-$. This is called a \emph{bad triangle} since we will have to violate at least one of these edges when clustering $\tilde{G}$. We must show that for $\alpha = 3$
	\begin{equation}
	\label{btin}
	w_{ij}^+ + w_{jk}^+ + w_{ik}^- \leq \alpha \left(   c_{ij} + c_{jk} + c_{ik} \right).
	\end{equation}
	Showing this inequality is somewhat tedious. The variables in \eqref{btin} are highly dependent on the types of edges shared among nodes $\{i,j,k\}$ in the original signed graph $G'$; there are two possibilities for each edge for a total of eight cases. We consider each case in turn. For notational simplicity we will write $ab^+$ if $(a,b) \in E^+$ and $ab^-$ if $(a,b) \in E^-$. We will repeatedly use the triangle inequality constraint satisfied by the variables, and the fact that $x_{ij} < 1/3$, $x_{jk} < 1/3$, and $1/3 \leq x_{ik} \leq x_{ij} + x_{jk} < 2/3$ by our construction of $F^+$ and $F^-$. \newline
	
	\noindent \textbf{Case 1: $(ij^+, jk^+, ik^-)$. }
	For this case $(c_{ij},c_{jk},c_{ik}) = ( (1-\lambda)x_{ij}, (1-\lambda)x_{jk}, \lambda (1-x_{ik}) )$ and $(w_{ij}^+,w_{jk}^+,w_{ik}^-) = (1-\lambda, 1-\lambda, \lambda)$, so
	\begin{align*}
	\alpha (c_{ij} &+ c_{jk} + c_{ik}) = 3 \left((1-\lambda)(x_{ij} + x_{jk}) + \lambda(1-x_{ik} )  \right) \\
	&\geq 3 \left( (1-\lambda)x_{ik} + \lambda(1-x_{ik}) \right) = 3\left( (1-2 \lambda) x_{ik} + \lambda \right) \\
	&> 3\left( (1-2\lambda) 2/3 + \lambda \right) = 2-\lambda = w_{ij}^+ + w_{jk}^+ + w_{ik}^-.
	\end{align*}
	We rely above on the fact that $(1-2\lambda) < 0$, which restricts our proof to cases where $\lambda > 1/2$.\newline
	
	\noindent \textbf{Case 2 $(ij^+, jk^-, ik^-)$; and Case 3: $(ij^-, jk^+, ik^-)$.} 
	
	\noindent For case 2, $(c_{ij},c_{jk},c_{ik}) = ( (1-\lambda)x_{ij}, \lambda (1-x_{jk}), \lambda (1-x_{ik}) )$ and $(w_{ij}^+,w_{jk}^+,w_{ik}^-) = (1-\lambda, 0, \lambda)$.
	Thus,
	\begin{align*}
	\alpha (c_{ij} & + c_{jk} + c_{ik}) = 3 \left( (1-\lambda) x_{ij} +\lambda(1-x_{jk}) + \lambda(1-x_{ik})\right)\\
	& \geq 3 \left( \lambda -\lambda x_{jk} + \lambda - \lambda x_{ik} \right) \geq 3 \left( \lambda - \lambda/3 + \lambda - 2\lambda/3 \right)\\
	&= 3\lambda \geq 1 = w_{ij}^+ + w_{jk}^+ + w_{ik}^-
	\end{align*}
	Case 3 is symmetric: switch the roles of edges $(i,j)$ and $(j,k)$ and the same result holds.\newline
	
	\noindent \textbf{Case 4: $(ij^-, jk^-, ik^-)$. }
	When all edges are negative the bound is loose since the left hand side of inequality~\eqref{btin} is $\lambda$ and we can easily bound the right hand side below:
	\begin{align*}
	\alpha (c_{ij} &+ c_{jk} + c_{ik}) = 3 (\lambda (1-x_{ij}) + \lambda (1-x_{jk}) + \lambda (1-x_{ik}))\\
	& = 3\lambda(3 - x_{ij} - x_{jk} - x_{ik}) > 3 \lambda (3 - 1/3 - 1/3 - 2/3) \\
	& = 3 \lambda ( 5/3) = 5 \lambda > \lambda = 0+0+\lambda = w_{ij}^+ + w_{jk}^+ + w_{ik}^-
	\end{align*}
	
	\noindent \textbf{Case 5: $(ij^+, jk^-, ik^+)$ and Case 6: $(ij^-,jk^+,ik^+)$.}
	A single proof works for both cases, using the fact that $(i,k) \in F^- \implies x_{ik} \geq 1/3$:
	\begin{align*}
	\alpha (c_{ij} &+ c_{jk} + c_{ik}) \geq 3(c_{ik}) = 3 (1-\lambda)x_{ik}\\
	&\geq 3 (1-\lambda) \frac13 = (1-\lambda) = w_{ij}^+ + w_{jk}^+ + w_{ik}^-.
	\end{align*}
	
	\noindent \textbf{Case 7: $(ij^-, jk^-, ik^+)$.}
	This case is trivial since $w_{ij}^+ + w_{jk}^+ + w_{ik}^- = 0$. 
	
	\noindent \textbf{Case 8: $(ij^+, jk^+, ik^+)$.}
	We apply the fact that $1/3 < x_{ik} \leq x_{ij} + x_{jk}$ to see:
	\begin{align*}
	\alpha (c_{ij} &+ c_{jk} + c_{ik})= 3(1-\lambda)(x_{ij} + x_{jk} + x_{ik})\\
	&\geq 3(1-\lambda)(2/3)= 2 (1-\lambda) = w_{ij}^+ + w_{jk}^+ + w_{ik}^-.
	\end{align*}
	For all cases we see that inequality~\eqref{btin} holds, so by Theorem~\ref{thm3pt1}, applying deterministic \alg{CC-Pivot} to $\tilde{G}$ will induce a clustering on $G'$ within a factor 3 of the LP lower bound. For uniformly random pivot nodes the same approximation holds in expectation.
\end{proof}

\subsection{2-Approximation for Cluster Deletion}
If we are explicitly interested in approximating the cluster deletion objective, then we alter \alg{threeLP} in two ways: we add the constraint $x_{ij} = 1$ for $(i,j) \in E^-$ to the LP-relaxation, and then we change how to round the output into an unweighted graph $\tilde{G} = (V,F^+,F^-)$ on which we apply \alg{CC-Pivot}. Pseudocode is shown in Algorithm \ref{alg:twocd}. 
\begin{algorithm}[tb]
	\caption{\twoCD}
	\begin{algorithmic}[5]
		\State{\bfseries Input:} Signed graph $G' = (V,E^+,E^-)$, $\lambda  \in (0,1)$
		\State {\bfseries Output:} Clustering~$\mathcal{C}$ of~$G'$
		\State Solve the LP-relaxation of cluster deletion:
		\begin{equation*}
		\begin{array}{ll} \text{minimize}  & \sum_{(i,j) \in E^+} x_{ij} \\ \subjectto  & x_{ij} \leq x_{ik} + x_{jk}  \text{ for all $i,j,k$} \\ & x_{ij} \in [0,1] \text{ for all (i,j) $\in E^+$}  \\ & x_{ij} = 1 \text{ for all (i,j) $\in E^-$}\end{array}
		\end{equation*}
		\State Define $\tilde{G} = (V,F^+,F^-)$ where
			\[ F^+ = \{(i,j) : x_{ij} < 1/2 \},\hspace{.5cm}   F^- = \{(i,j) : x_{ij} \geq 1/2 \}  \]
		\State Return \alg{CC-Pivot}($\tilde{G}$)
	\end{algorithmic}
	\label{alg:twocd}
\end{algorithm}


We name the resulting procedure \alg{twoCD}, and prove the following result:
\begin{theorem}
	\label{CD2}
	Algorithm \alg{twoCD} satisfies Theorem~\ref{thm3pt1} with $\alpha = 2$.
\end{theorem}
\begin{proof}
	First observe that no negative edge mistakes are made by performing \alg{CC-Pivot} on $\tilde{G}$: if $k$ is the pivot and $i,j$ are two positive neighbors of $k$ in $\tilde{G}$, then $x_{ik} < 1/2$, $x_{jk} < 1/2$, and $x_{ik} \leq x_{ik} + x_{jk} < 1$. Since all distances are less than one, all nodes must share positive edges in $G'$.
	Now we show that the assumptions of Theorem~\ref{thm3pt1} are satisfied.
	Recall that for cluster deletion the edge weights for $(i,j) \in E^+$ are $(w_{ij}^+,w_{ij}^-) = (1,0)$ and for $(i,j) \in E^-$ we have $(w_{ij}^+,w_{ij}^-) = (0,1)$. The LP budget for $(i,j) \in E^+$ is $x_{ij}$ and is zero otherwise.
	
	For part (1) of Theorem~\ref{thm3pt1}, note that $(i,j) \in F^+ \implies (i,j) \in E^+ \implies w_{ij}^- = 0 \leq 2c_{ij}.$ If $(i,j) \in F^-\cap E^-$, we have $w_{ij}^+ = 0$ and the inequality $w_{ij}^+ \leq \alpha c_{ij}$ is trivial. Finally, if $(i,j) \in F^- \cap E^+$, then $x_{ij} \geq 1/2$, $w_{ij}^+ = 1$, and $c_{ij} = x_{ij}$, so the inequality $w_{ij}^+ = 1 = 2( 1/2) \leq 2 x_{ij} = \alpha c_{ij}$ holds.
	
	For part (2) of Theorem~\ref{thm3pt1}, let $\{i,j,k\}$ be a bad triangle in $\tilde{G} = (V,F^+,F^-)$ where $(i,k) \in F^-$ is the negative edge. The key is to notice that $\{i,j,k \}$ must form a triangle of all positive edges in the original signed $G'$. Since $\{(i,j), (j,k)\} \subset F^+$, we see $x_{ij}<1/2$, $x_{jk}<1/2$, and therefore $x_{ik} \leq x_{ij} + x_{jk} < 1$. Since these distances are strictly less than one, all of the edges are positive in $G'$. This means that $(c_{ij}, c_{jk}, c_{ik}) = (x_{ij}, x_{jk}, x_{ik})$ and $(w_{ij}^+, w_{jk}^+, w_{ik}^-) = (1,1,0)$. Also $(i,k) \in F^-$ implies $x_{ik} \geq 1/2$, and combining these facts with the triangle inequality yields the desired result:
	\begin{align*}
	\alpha \left( c_{ij} + c_{jk} + c_{ik} \right) &= 2\left(x_{ij} + x_{jk} + x_{ik} \right) \geq 2 (2 x_{ik}) \geq 2 = w_{ij}^+ + w_{jk}^+ + w_{ik}^- .
	\end{align*}
	Therefore, Theorem~\ref{thm3pt1} guarantees that \alg{CC-Pivot} will output a clustering that at most costs $2 \sum_{i<j} c_{ij} = 2 \sum_{(i,j) \in E^+} x_{ij}$, so this is a two-approximation for cluster deletion.
\end{proof}
This result is particularly interesting given that no constant-factor for cluster deletion has been explicitly presented in previous literature. In contrast, numerous approximation results have been presented for unweighted correlation clustering, culminating in the 2.06 approximation given by Chawla et al.~\cite{chawla2015near}. Our result indicates for the first time that although far fewer approximation algorithms for cluster deletion have been developed, it is in a sense an easier problem to approximate than correlation clustering. 

\subsection{5-Approximation for \lcc}
For completeness we also include details for a 5-approximation for \lcc when $\lambda >1/2$, which also relies on solving the LP-relaxation~\eqref{eq:cclp}. The rounding scheme and proof technique of this method are similar to those developed by Charikar et al.\ for~$\pm 1$
correlation clustering~\cite{charikar2005clustering}. \newln{The rounding scheme we apply here is in fact the same as the approach developed earlier and independently by Puleo and Milenkovic for a specially weighted version of correlation clustering that can be viewed as a generalization of standard \lcc when $\lambda > 1/2$~\cite{puleo2015correlation}. We give proof details here specifically for \lcc weights.} Pseudocode for the method is given in Algorithm~\ref{alg:LP5}. We refer to this as \fiveLP, based on the following approximation result.
\begin{algorithm}[tb]
	\caption{\alg{fiveLP}}
	\begin{algorithmic}[5]
		\State{\bfseries Input:} Signed graph $G' = (V,E^+,E^-)$, $\lambda  \in (0,1)$
		\State {\bfseries Output:} Clustering~$\mathcal{C}$ of~$G'$
		\State Solve the LP-relaxation of ILPs~(\ref{eq:cc},\ref{lamobjstand}), obtaining \emph{distances}~$(x_{ij})$
		\State $W \gets V$, $\mathcal{C} \gets \varnothing$
		\While{$W\neq \varnothing$ }
		\State Choose $u \in W$ arbitrarily
		\State $T \gets \{i \in W: x_{ui} \leq 2/5 \}$
		\If {average \emph{distance} between~$u$ and~$T$ is $< {1}/{5}$}
		\State $S := \{u\} \cup T$
		\Else 
		\State $S := \{u\}$
		\EndIf 
		\State $\mathcal{C} \gets \mathcal{C} \cup \{S \}$, $W \gets W\backslash S $
		\EndWhile
	\end{algorithmic}
	\label{alg:LP5}
\end{algorithm}
\begin{theorem}
	\label{thm:5lp}
	Algorithm~\fiveLP
	gives a factor-$5$ approximation for \lcc for all $\lambda > {1}/{2}$.
\end{theorem}
\begin{proof}
We prove the 5-approximation holds for objective~\eqref{alpobj} when $\alpha = \lambda/(1-\lambda)$, since this objective is equivalent to \lcc in terms of approximations. In other words,
we are considering the relaxation of a correlation clustering problem where each positive edge has weight~$1$ and each negative
edge has weight $\alpha > 1$.
Solving this LP gives a lower bound on the optimal \lcc score. We show that both for singleton and non-singleton clusters formed by \fiveLP, the
number of mistakes made at each cluster is within a factor five of the LP cost corresponding to that cluster. Recall that each cluster is formed around some node~$u$, and $T = \{i \in W: x_{ui} \leq 2/5 \}$.

\paragraph{Singleton Clusters.}
If~$u$ is a singleton, we know that~$\sum_{i\in T} x_{ui} \geq {|T|}/5$. In this case we make at most~$|T|$ mistakes,
which would happen if all edges between~$u$ and~$T$ are positive. Given that $x_{ui} \leq 2/5$ for every~$i \in T$, we know that $(1-x_{ui}) \geq {3}/{5} \geq x_{ui}$.
Let $T^+ = \{i \in T: (u,i) \in E^+ \}$ and $T^- = \{i \in T: (u,i) \in E^- \}$. The LP cost associated with this cluster is
\[ \sum_{i\in T^+} x_{ui} + \alpha \sum_{i\in T^-} (1-x_{ui}) > \sum_{i\in T^+} x_{ui} + \sum_{i\in T^-} x_{ui}
\geq \sum_{i\in T} x_{ui}  \geq \frac{|T|}{5}\,, \]
so we account for the errors within a factor~$5$.

\paragraph{Negative-edge mistakes in non-singleton clusters.}
Consider a negative edge inside a cluster of the form $S = \{u\} \cup T$. If that edge is $(u,j)$, the LP cost is $1-x_{uj} \geq 1- 2/5 = 3/5$. For
every other negative edge, $(i,j)$, where $u \notin \{i,j\}$,
the LP cost is $1-x_{ij} \geq 1 - x_{ui} - x_{uj}\geq 1- 4/5 = 1/5$. Either way, the LP has paid at least~$1/5$ for each negative-edge mistake.

\paragraph{Positive-edge mistakes in non-singleton clusters.}
Positive edges from~$u$ to~$j \notin T$ satisfy $x_{uj} > 2/5$, so this type of edge pays for itself easily.
The other edges we need to account for are all edges $(i,j) \in E^+$ where~$i \in T$ and~$j \notin T$.
We will charge all edges of this form to the node~$j$ that lies outside~$T$. 

First, if~$x_{uj} \geq 3/5$, then
$x_{ij} \geq x_{uj} - x_{ui} \geq 3/5 - 2/5 = 1/5$ and
the positive edge pays for itself within factor~$5$.

Now, fix some $j \notin T$ where $x_{uj} \in \left( {2}/{5}, {3}/{5} \right)$.
Let $T^+_j = \{i \in T: (i,j) \in E^+ \}$ and $T^-_j = \{i \in T: (i,j) \in E^- \}$, and let $p_j = |T^+_j|$
be the number of positive~$(T,j)$ edges, while $n_j = |T^-_j|$ is the number of negative~$(T,j)$ edges.
The number of positive mistakes we are charging to~$j$ is exactly~$p_j$, and we have
\begin{align*}
\text{LP cost at $j$} &= {\textstyle {\sum_{i \in T^+_j} }x_{ij} + \alpha { \sum_{i \in T^-_j} }(1-x_{ij}) }\\
&\geq {\textstyle \sum_{i \in T^+_j} (x_{uj} - x_{ui}) + \sum_{i \in T^-_j} (1-x_{uj}-x_{ui}) }\\
&= {\textstyle p_j x_{uj} + n_j(1-x_{uj}) - \sum_{i\in T} x_{ui} }\\
&> {\textstyle p_j x_{uj} + n_j(1-x_{uj}) - ({p_j + n_j})/{5}\,, }
\end{align*}
where the last inequality follows from the fact that the average distance from~$u$ to~$T$ is less than~$1/5$.
Thus the LP cost is bounded by a linear function, $p_j(x_{uj}-1/5) + n_j(4/5-x_{uj})$,
where $x_{uj} \in \left( {2}/{5}, {3}/{5} \right)$.
Hence the coefficient of~$p_j \geq 1/5$, while the coefficient of~$n_j \geq 0$,
so~$p_j$ is within a factor~$5$ of the LP cost. 
\end{proof}


\subsection{4-approximation for Cluster Deletion.} We slightly alter \fiveLP in the following ways whenever~$\lambda > m/(m+1)$:
\begin{itemize}
	\item For all $(i,j) \notin E$, force constraints~$x_{ij} = 1$ in the LP.
	\item When rounding, select arbitrary $u \in W$ and set $T \gets \{i \in W: x_{ui} < 1/2 \}$ (rather than using $x_{ui} \leq 2/5$).
	\item Make $u$ a singleton if the average distance from $T$ to $u$ is $\leq 1/4$, otherwise cluster $u$ with $T$.
\end{itemize}
\newln{Charikar et al.\ have already noted in previous work that this algorithm will produce a 4-approximation for cluster deletion~\cite{charikar2005clustering}. We call this algorithm~\alg{fourCD}, and for completeness include a proof of the approximation guarantee.}
\begin{theorem}
	Algorithm \alg{fourCD} returns a 4-approximation to cluster deletion.
\end{theorem}
\begin{proof}
First, \alg{fourCD} forms only cliques. Should the cluster formed around~$u$ not be a singleton,
for every $i,j \in T$, with~$i\neq j$, we know $x_{ui},x_{uj} < {1}/{2}$, so
$x_{ij} \leq x_{ui} + x_{uj} < {1}/{2} + {1}/{2} = 1$.
Since the distance $x_{ij}$ is strictly less than~$1$, nodes $i,j$ must be adjacent, or else we would have forced
$x_{ij} = 1$ in the LP-relaxation. 
We therefore only need to account for positive-edge mistakes.
The remainder of the proof follows directly from the same steps used to prove Theorem~\ref{thm:5lp},
as well as the original proof of Charikar et al.~\cite{charikar2005clustering} for~$\pm 1$ correlation clustering:
\paragraph{Singleton Clusters}
If we cluster~$u$ as a singleton, then the number of mistakes we make between~$u$ and~$T$ is exactly~$|T|$, as these are all positive neighbors of~$u$.
Since~$u$ was made a singleton cluster, we know that $\sum_{i\in T} x_{ui} \geq |T|/4$, so these positive mistakes are paid within factor four.
Finally, note that every positive edge $(u,j) \in E^+$ for $j \notin T$ has LP cost greater than~$1/2$, so those mistakes are paid for within factor
two.

\paragraph{Clusters $S = \{u\}\cup T$}
No negative mistakes are made, so we only need to account for positive mistakes.
As mentioned in the previous case, edges $(u,j) \in E^+$ for $j \notin T$ pay for themselves within factor two.
For $j \notin T$ where $x_{uj} \geq 3/4$, if $(i,j) \in E^+$ and $i \in T$, we know $x_{ij} \geq x_{uj} - x_{ui} > 3/4 - 1/2 = 1/4$, so the edge
pays for itself within factor four.

Now consider a single node $j \notin T$ such that $x_{uj} \in [1/2,3/4)$, and then consider all $i \in T$ with $(i,j) \in E^+$.
Again, use the notation $T^+_j = \{i \in T: (i,j) \in E^+ \}$ and $T^-_j = \{i \in T: (i,j) \in E^- \},$ with $p_j = |T^+_j|$ and $n_j = |T^-_j|$.
Now we bound the weight of positive mistakes as a function of the LP cost associated with~$j$.
Thanks to the constraint $x_{ij} = 1$ for all $\nij$, $\sum_{i \in T^-_j} (1-x_{ij}) = 0$, therefore, relying also on the reasoning for~\fiveLP:
\begin{align*}
\sum_{i \in T^+_j} x_{ij} &= \sum_{i \in T^+_j} x_{ij} + \sum_{i \in T^-_j} (1-x_{ij}) \\
&\geq p_j x_{uj} + n_j(1-x_{uj}) - \sum_{i\in T} x_{ui} \\
&\geq p_j x_{uj} + n_j(1-x_{uj}) - \frac{p_j + n_j}{4} \\
&=p_j\left(x_{uj} - \frac 1 4\right) + n_j\left(\frac 3 4 - x_{uj}\right)\,.
\end{align*}
For $x_{uj} \in [1/2,3/4)$,
the coefficient of~$p_j \geq 1/4$, while the coefficient of~$n_j \geq 0$.
Therefore the number of mistakes,~$p_j$, is paid for within factor four.
\end{proof}	

\subsection{Scalable Heuristic Algorithms}
As a counterpart to the previous approximation-driven approaches,
we provide fast algorithms for \lcc-based greedy local heuristics.
The first of these is \alg{GrowCluster} (Algorithm~\ref{alg:grow}), which iteratively selects an
unclustered node uniformly at random and forms a cluster around it by greedily aggregating adjacent nodes, until there is no more improvement to the \lcc objective.

A variant of this, called \alg{GrowClique} (Algorithm~\ref{alg:gclique}), is specifically designed for cluster deletion.
It monotonically improves the \lcc objective, but differs in that at each iteration
it randomly selects~$k$ unclustered nodes, and greedily grows cliques around each of these seeds.
The resulting cliques may overlap: at each iteration we select only the largest of such cliques.

\begin{algorithm}[tb]
	\caption{\alg{GrowCluster}}
	\label{alg:grow}
	\begin{algorithmic}[5]
		\State{\bfseries Input:} $G' = (V,E^+,E^-)$
		\State {\bfseries Output:} a clustering~$\mathcal{C}$ of~$G'$
		\State $W \leftarrow V$, $C \leftarrow \emptyset$
		\While {$W\neq \emptyset$ }
		\State 1. Choose a uniformly random $u \in W$, set $S \leftarrow \{u\}$
		\State 2. For all $v \in W\backslash S$, compute benefit from merging $v$ into $S$:
		\State \hspace{1cm} (For standard \lcc: $\Delta_v = \cut(S,\{v\}) - \lambda |S|$)
		\State  \hspace{1cm} (For degree-weighted: $\Delta_v = \cut(S,\{v\}) - \lambda w_v \vol(S)$
		\State 3. Set $m \leftarrow \max_v \Delta_v$, $v' \leftarrow \arg \max \Delta_v$
		\While {$m>0$}
		\State $S \leftarrow S\cup \{v'\}$
		\State Update $\Delta_v$,~$m$, and~$v'$
		\EndWhile
		\State Add cluster~$S$ to~$\mathcal{C}$, update $W \leftarrow W\backslash S$
		\EndWhile
	\end{algorithmic}
\end{algorithm}

\begin{algorithm}[tb]
	\caption{\alg{GrowClique}}
	\label{alg:gclique}
	\begin{algorithmic}[5]
		\State{\bfseries Input:} $G' = (V,E^+,E^-)$
		\State {\bfseries Output:} a clustering~$\mathcal{C}$ of~$G'$ where all clusters are cliques
		\State $W \gets V$, $C \gets \emptyset$
		\While {$W\neq \emptyset$ }
		\For{$i = 1$ to $k$}
		\State Select a random seed node $u \in W$, set $S \gets \{u\}$
		\State Set $N \gets \{v \in W\backslash S : \mbox{$v$ neighbors all nodes in S}\}$
		\While{ $N \neq \emptyset$ }
		\State $S \gets S \cup \{v\}$ for any $v \in N$
		\State $N \gets \{v \in W\backslash S : \mbox{$v$ neighbors all nodes in S}\}$
		\EndWhile
		\State $S_i \gets S$
		\EndFor
		\State $S_{\max} = \argmax_i |S_i|$
		\State Add cluster~$S_{\max}$ to~$\mathcal{C}$, update $W \gets W\backslash S_{\max}$
		\EndWhile
	\end{algorithmic}
\end{algorithm}

Finally, since the \lcc and Hamiltonian objectives are equivalent,  we can use previously developed
algorithms and software for modularity-like objectives with a resolution parameter. In particular we employ
adaptations of the \emph{Louvain} method, an algorithm developed by Blondel et al.~\cite{blondel2008louvain}.
It iteratively
visits each node in the graph and moves it to an adjacent cluster, if such a move gives a locally maximum increase in
the modularity score.
This continues until no move increases modularity, at which point the clusters are aggregated into
super-nodes and the entire process is repeated on the aggregated network. By adapting the original Louvain method to make
greedy local moves based on the \lcc objective, rather than modularity, we obtain a scalable algorithm that is known to
provide good approximations for a related objective, and additionally adapts well to changes in our parameter~$\lambda$.
We refer to this as \alg{Lambda-Louvain}.
Both standard and degree-weighted versions of the algorithm can be achieved by
employing existing generalized Louvain algorithms (e.g., the GenLouvain algorithm of Jeub et al. \url{http://netwiki.amath.unc.edu/GenLouvain/}).

{Regarding the scalability of these algorithms, we note that in practice for all of these methods we do not explicitly form the signed graph $G'$, which in theory has $O(n^2)$ (positive or negative) edges. Given an initial sparse graph $G$, it suffices to store positive-edge relationships between nodes, and implicitly apply penalties due to negative edges in $G'$ by considering non-edges in $G$}. 
Our heuristic algorithms satisfy the following guarantee:
\begin{theorem}
	\label{thm:sscbound}
	For every~$\lambda$, standard (respectively, degree-weighted) \alg{Lambda-Louvain} either places all nodes in one cluster, or produce clusters that have scaled sparsest cut (respectively, scaled normalized cut) bounded above by~$\lambda$. The same holds true for \alg{GrowCluster}.
\end{theorem}

\begin{proof}
	Note that by design, when \alg{Lambda-Louvain} terminates there will be no two clusters which can be merged to yield a better objective score. Just as in the proof of  statement (1) for Theorem~\ref{thm:lambound}, for the standard \lcc objective this means that for any pair of cluster $S_i$ and $S_j$ we have
	\begin{equation}
	\label{true}
	 \cut(S_i,S_j) - \lambda|S_i||S_j| \leq 0\,. 
	\end{equation}
	We then fix~$S_i$, perform a sum over all other clusters, and get the desired result:
		\[ { \sum_{j \neq i}} \cut(S_i,S_j) - { \sum_{j \neq i}} \lambda |S_i||S_j| \leq 0\, \implies \cut(S_i, \bar{S_i}) - \lambda |S_i||\bar{S_i}| \leq 0 \implies \frac{\cut(S_i, \bar{S_i})}{|S_i||\bar{S_i}|}
		\leq \lambda\, . \]
	If we are using degree-weighted \alg{Lambda-Louvain}, when the algorithm terminates we know that all pairs of clusters $S_i, S_j$ satisfy
		\begin{equation*}
		\cut(S_i,S_j) - \lambda\vol(S_i)\vol(S_j) \leq 0\,
		\end{equation*}
	and the corresponding result for scaled normalized cut holds.
	
	Though slightly less obvious, it is also true that none of the output clusters of \alg{GrowCluster} (if there are at least two) could be merged to yield a better objective score. Notice that this is certainly true of the first cluster $S_1$ formed by \alg{GrowCluster}: we stop growing $S_1$ when we find that (for standard-weighted \lcc)
	\[\cut(S_1, v) - \lambda |S_1| \leq 0\]
	for all other nodes $v$ in the graph. Therefore, given \emph{any} other subset of nodes $S$ (including sets of nodes making up other clusters that the algorithm will output), we see
	\[ \sum_{v \in S} \left( \cut(S_1, v) - \lambda |S_1| \right)= \cut(S_1,S) - \lambda |S_1| |S| \leq 0. \]
	Therefore when we form the second cluster $S_2$ with \alg{GrowCluster}, we already know that $\cut(S_1,S_2) - \lambda |S_1| |S_2| \leq 0$, and similar reasoning shows that $\cut(S_2, S_j) - \lambda |S_2| |S_j| \leq 0$ will hold for any cluster $S_j$ with $j >2$ that will be subsequently formed. In this way we see that inequality~\eqref{true} will also hold between all pairs of clusters output by \alg{GrowCluster}, so the rest of the result follows. The same steps will also work for degree-weighted \lcc.
\end{proof}

\section{Experiments}
We begin by comparing our new methods against existing correlation clustering algorithms on several small networks. This shows our algorithms for \lcc are superior to common alternatives.
We then study how well-known graph partitioning algorithms implicitly optimize the \lcc objective for various~$\lambda$.  In subsequent experiments,
we apply our methods to clique detection in collaboration and gene networks, and to social network analysis. 


\subsection{\lcc on Small Networks}
In our first experiment, we show that \alg{Lambda-Louvain} is the best general-purpose correlation clustering method for minimizing the \lcc objective. We test this on four small networks: Karate~\cite{karate}, Les Mis~\cite{lesmis},
Polbooks~\cite{polbooks}, and Football~\cite{girvan2002community}.
Figure~\ref{smallnets} shows the performance of our algorithms, \alg{Pivot}, and~\alg{ICM}, for a range of~$\lambda$
values. \alg{Pivot} is the fast algorithm of Ailon et al.~\cite{ailon2008aggregating}, which selects a uniform random node and clusters its neighbors
with it.
\alg{ICM} is the energy-minimization heuristic algorithm of Bagon and Galun~\cite{bagon2011large}. 

\begin{figure}[t]
	\centering
	\subfloat[Karate, $n = 34$, $\lambda^* = .0267$\label{fig:karate}]
	{\includegraphics[width=.35\linewidth]{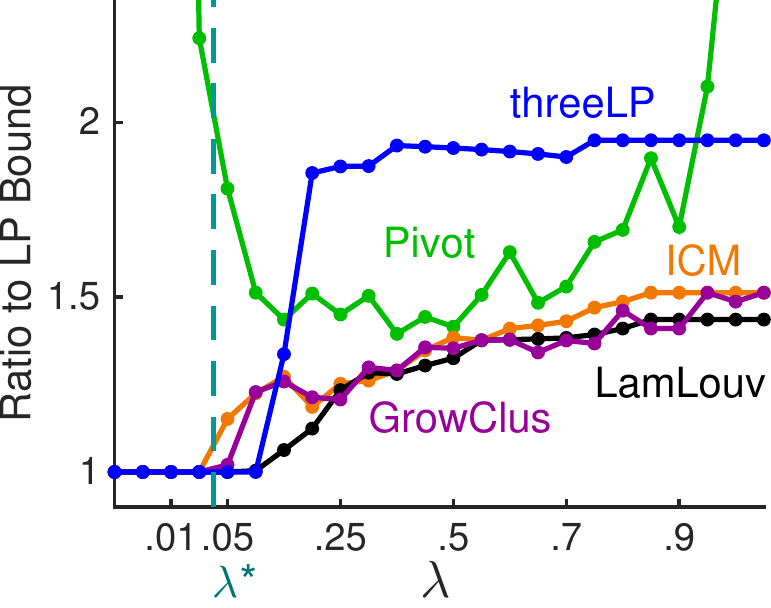}}\hspace{2cm}
	\subfloat[Les Mis, $n = 77$, $\lambda^* = .0045$ \label{fig:lesmis}]
	{\includegraphics[width=.35\linewidth]{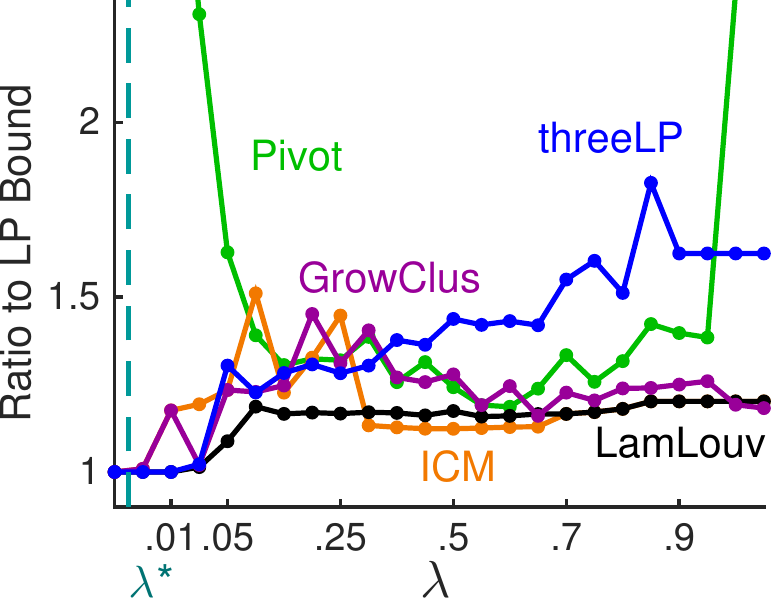}}\hspace{2cm}
	\centering
	\subfloat[Polbooks, $n = 105$, $\lambda^* = .0069$\label{polbooks}]
	{\includegraphics[width=.35\linewidth]{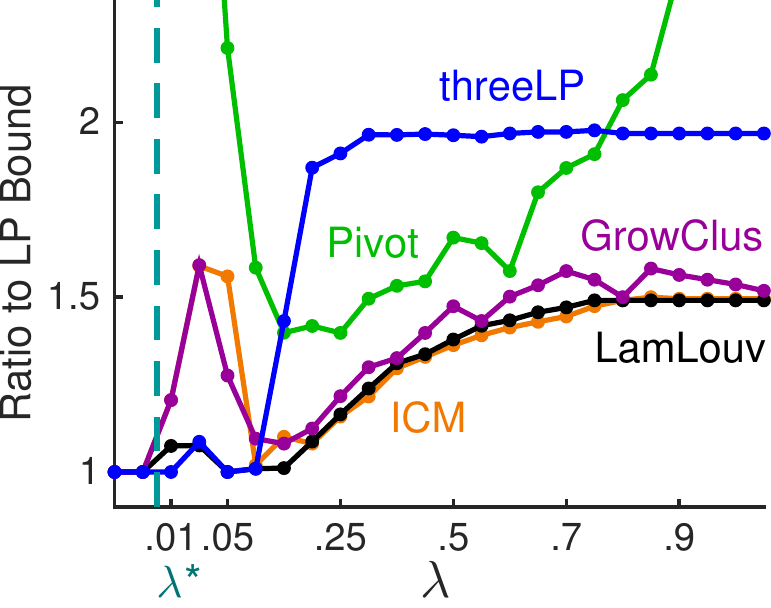}}\hspace{2cm}
	\subfloat[Football, $n = 115$, $\lambda^* = .0184$\label{fig:football}]
	{\includegraphics[width=.35\linewidth]{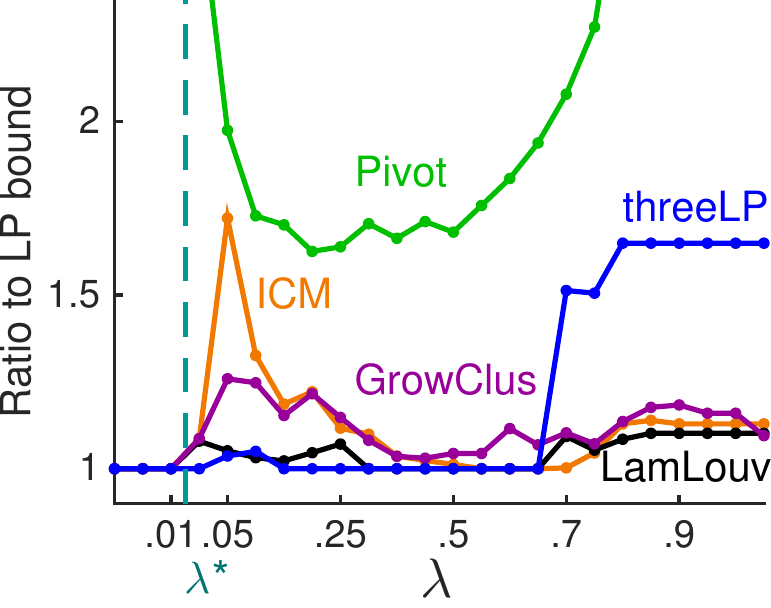}}
	\caption{We optimize the standard \lcc objective with five correlation clustering algorithms on four small networks. The~$y$-axis reports the ratio between each algorithm's score and the lower bound on the optimal objective determined by solving the LP relaxation. 
		\threeLP (blue) gives much better than a factor-$3$ approximation in practice, and performs especially well for small values of~$\lambda$. \alg{ICM} (orange) is faster and gives good approximations. \alg{Pivot} (green) is fast, but does very poorly for extreme values of~$\lambda$ (near either~$0$ or~$1$). 
		\alg{Lambda-Louvain} (black) and \alg{GrowCluster} (violet) perform well for all~$\lambda$ in addition to being the most scalable algorithms.
		In each plot, a dashed vertical line indicates the optimal scaled sparsest cut value,~$\lambda^*$, for that network.}
	\label{smallnets}
\end{figure}

%
We find that~\fiveLP gives much better than a $5$-approximation in practice. \alg{Pivot} is much faster, but performs poorly for~$\lambda$ close to zero or one. \alg{ICM} is also much quicker than solving the LP relaxation, but is still limited in scalability as it is intended for correlation clustering problems where most edge weights are zero, which is not the case for \lccb. On the other hand, \alg{GrowCluster} and \alg{Lambda-Louvain} are scalable and give good approximations for all input networks and values of~$\lambda$.

\subsection{Standard Clustering Algorithms}
Many existing clustering algorithms implicitly optimize different parameter regimes of the \lcc objective.
We show this by running several clustering algorithms on two graphs, (1) a 1000-node synthetic graph generated from the BTER model~\cite{sesh2012bter}, and (2) the largest component (4158 nodes) of the ca-GrQc collaboration network from the \emph{arXiv} e-print website.
We cluster each graph using Graclus~\cite{dhillon2007weighted} (forming two clusters), Infomap~\cite{bohlin2014community},
and Louvain~\cite{blondel2008louvain}. To
form dense clusterings, we also partition the networks by recursively extracting the maximum clique
(called \alg{RMC}), and by recursively extracting the maximum quasi-clique (\alg{RMQC}), i.e., the largest set of nodes with inner edge density bounded below by some~$\rho < 1$. The last two procedures must solve an NP-hard objective at each step,
but for reasonably sized graphs there is available clique and quasi-clique detection software~\cite{rossi2015maxclique,liu2008effective}. 

After each algorithm has produced a single clustering of the unsigned network, we evaluate how the clustering's \lcc objective score changes as we vary $\lambda$. This allows us to observe whether the clustering produced by an algorithm is effectively approximating the optimal \lcc objective for a certain choice of $\lambda$. In Figure~\ref{clusalgs} we report for each $\lambda$ the ratio between each clustering's objective score and the \lcc LP-relaxation lower bound. We compare against running \alg{Lambda-Louvain}, which produces a different clustering for each value of $\lambda$. We also display adjusted rand index (ARI) scores between the \alg{Lambda-Louvain} clustering and the output of other algorithms. We note that the ARI scores peak in the same regime where each algorithm best optimizes \lcc. Typically the ARI peaks are higher for larger $\lambda$. This can be explained by realizing that when $\lambda$ is small, fewer clusters are formed. It is natural to expect there to be many ways to partition the graph into a small number of clusters such that different clusterings share a very similar structure, even if the individual clusters themselves do not match. On the whole, the plots in Figure~\ref{clusalgs} illustrate that our framework and algorithm effectively interpolate between several well-established strategies in graph partitioning, and can serve as a good proxy for any clustering task for which any one of these algorithms is known to be effective. 

\begin{figure}[t]
	\centering
	\subfloat[BTER graph with 1000 nodes\label{btr}]
	{\includegraphics[width=.35\linewidth]{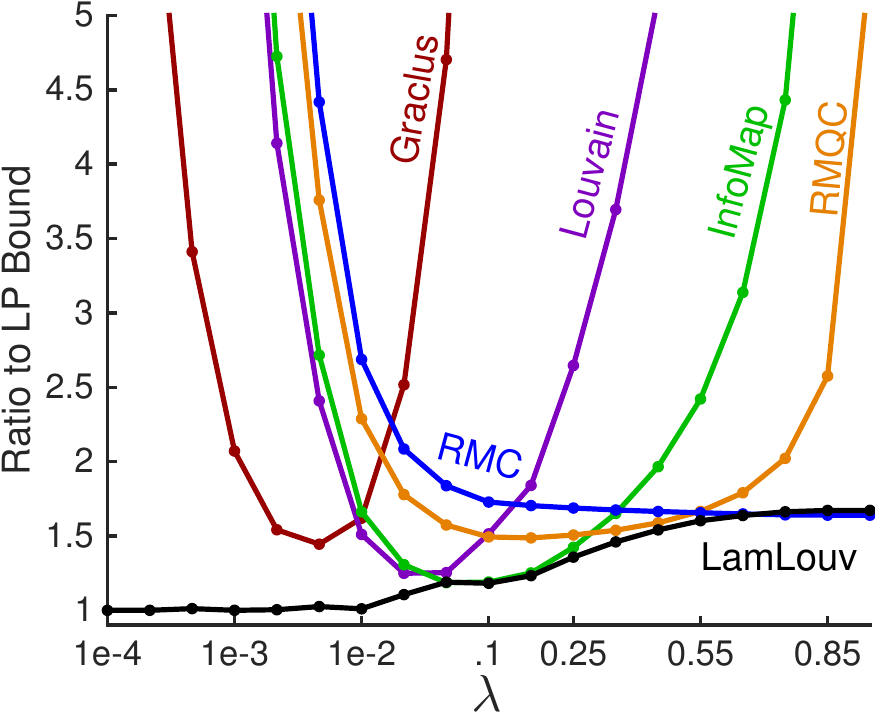}}\hspace{2cm}
	\subfloat[Largest component of ca-GrQc \label{cagrqc}]
	{\includegraphics[width=.35\linewidth]{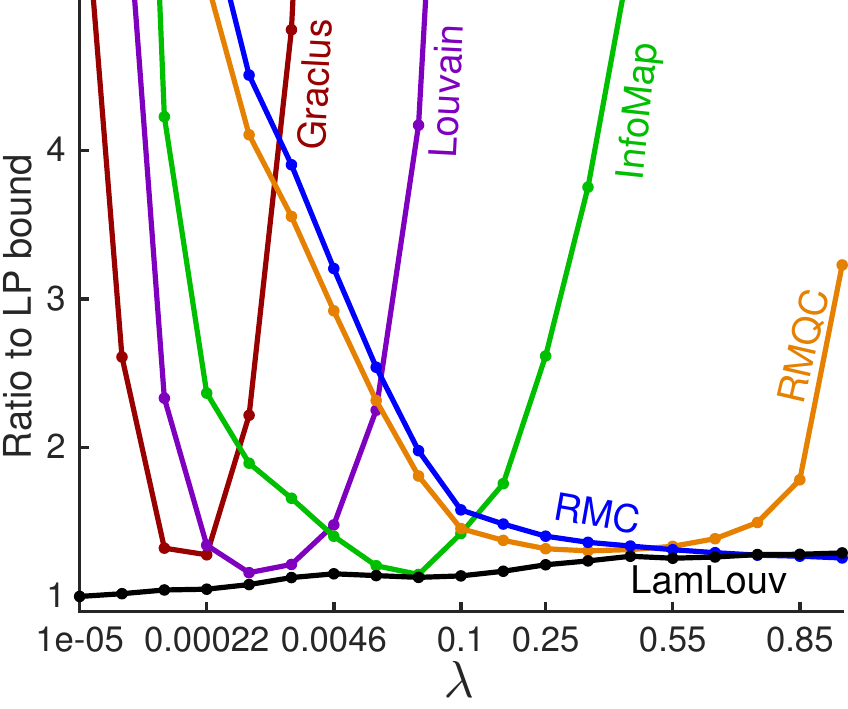}}\hspace{2cm}
	\centering
	\subfloat[ARI scores (BTER)\label{btrari}]
	{\includegraphics[width=.35\linewidth]{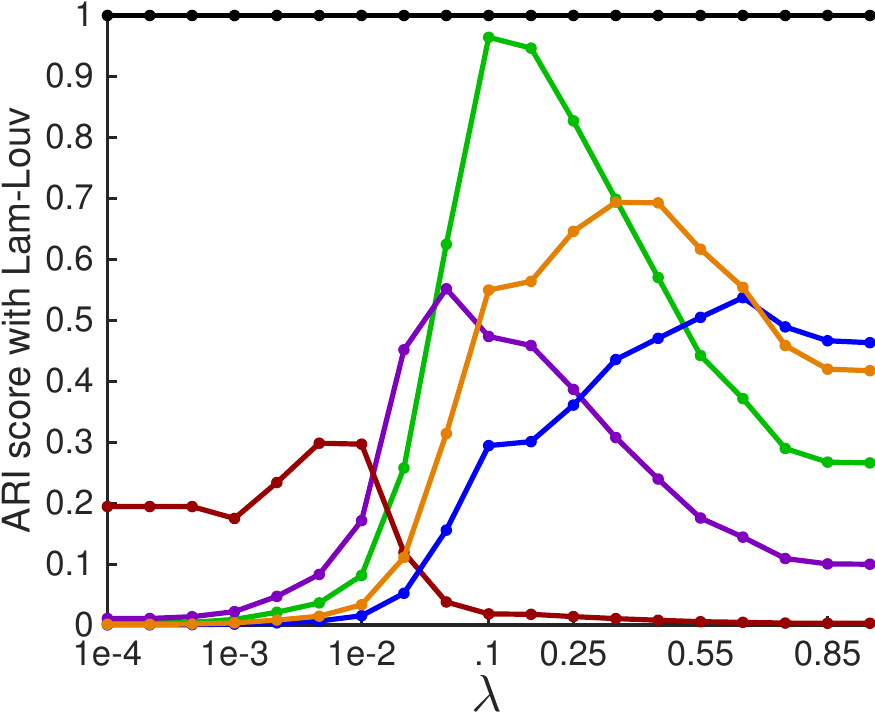}\label{bter}}\hspace{2cm}
	\subfloat[ARI scores (ca-GrQc) \label{cagrqcari}]
	{\includegraphics[width=.35\linewidth]{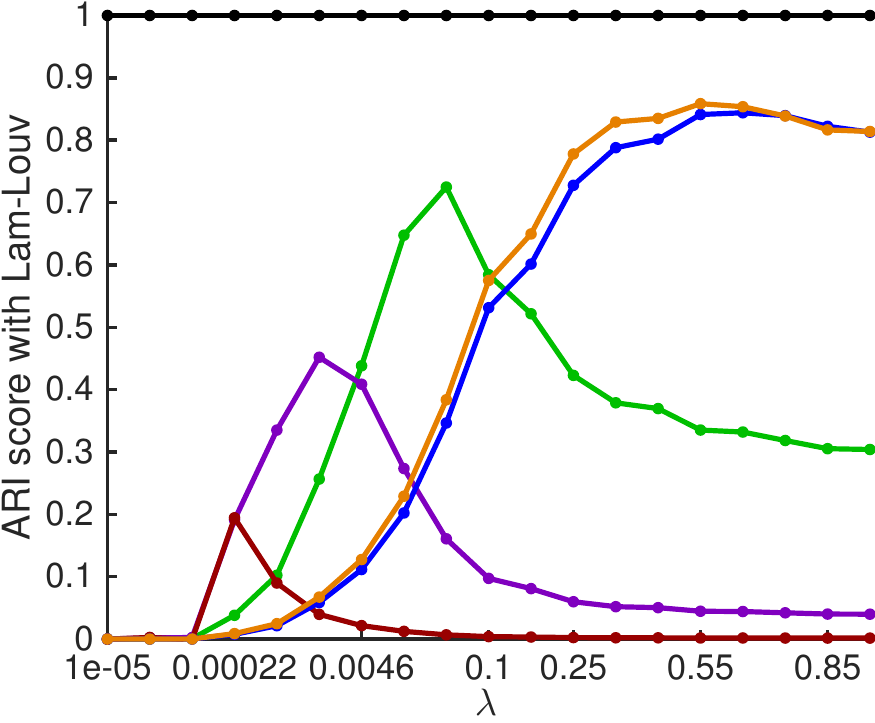}}
	\caption{On top we illustrate the performance of well-known clustering algorithms in approximating the \lcc objective
		on (a) one synthetic and (b) and one real-world graph. 
		The bowl-shaped curves indicate that each algorithm implicitly optimizes the \lcc objective in a different parameter regime. The $y$-axis reports the ratio between each clustering's objective score and the LP-relaxation lower bound. \alg{Lambda-Louvain} interpolates between all the clustering strategies seen here, always giving an approximation ratio~$\leq 2$. 
		In the lower plots we show the ARI score between each clustering and \alg{Lambda-Louvain} for both the BTER graph (c) and ca-GrQc (d). These show peaks in the same parameter regime where each algorithm is most successful at approximating \alg{LambdaCC}. }
	\label{clusalgs}
\end{figure}

By performing multiple runs of Graclus and varying the number of partitions formed by this algorithm, we can show that Graclus can approximately optimize different
parameter regimes of \alg{LambdaCC}. In Figure~\ref{manygraclus} we show how the Graclus objective scores change as we increase the number of clusters
from~$2$ to over~$2000$. As the number of clusters increases, the algorithm performs better and better for large~$\lambda$ and worse for smaller~$\lambda$.
Figure~\ref{manyrmqc} shows that something similar occurs for RMQC when we vary the minimum density of quasi-cliques from~$0.5$ to~$0.85$. As the inner-edge density increases, the performance of RMQC essentially converges to the performance of RMC. 

\paragraph{Solving the LP relaxation}
{Our results in this experiment involve solving the LP-relaxation of correlation clustering, which includes~$\Theta(n^3)$ triangle inequality constraints, for graphs of size $n = 1000$ and $n = 4158$. For these problems the LP constraint matrix is extremely large, and standard black-box LP solvers are impractical, since even forming the constraint matrix is prohibitively expensive. We employ two general strategies for overcoming the huge memory requirement in practice. The first is to solve the LP on a subset of the constraints, then iteratively update the constraint set and re-solve the LP as needed, until convergence. The second approach employs the triangle-fixing procedure of Dhillon et al.\ for the related metric nearness problem~\cite{dhillon2005triangle}. }

\begin{figure}
	\centering
	\subfloat[Multiple runs of Graclus on the caGrQc network, with increasingly many clusters\label{manygraclus}]
	{\includegraphics[width=.35\linewidth]{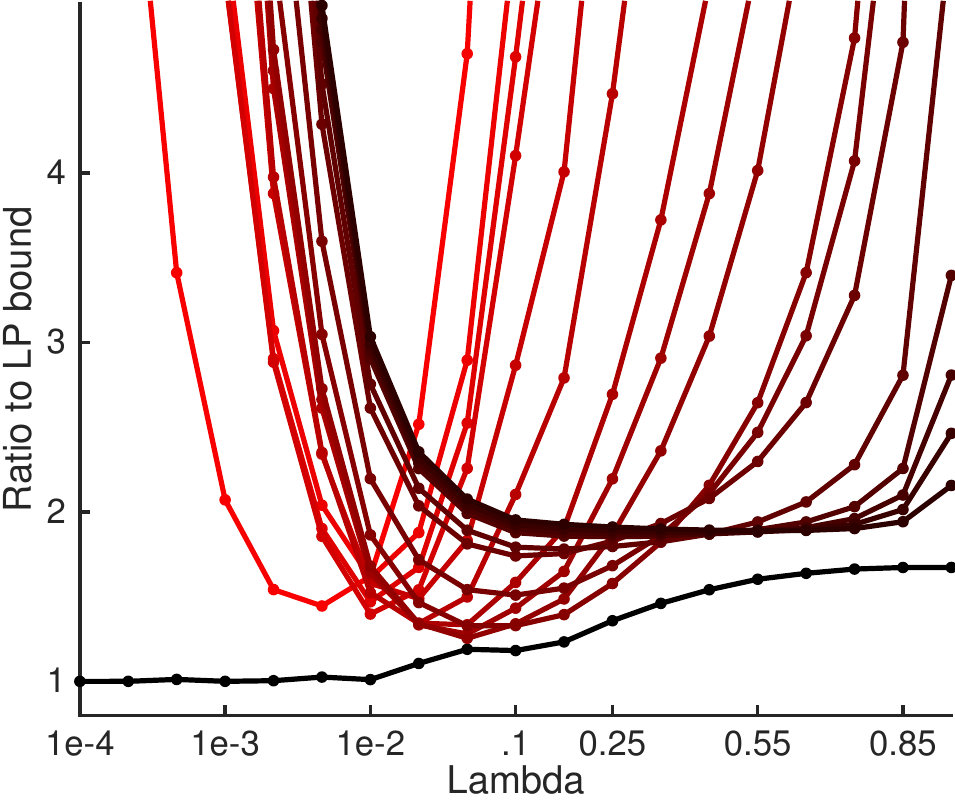}}\hspace{2cm}
	\subfloat[Multiple runs of RMQC on the caGrQc network, with increasingly higher quasi-clique density\label{manyrmqc}]
	{\includegraphics[width=.35\linewidth]{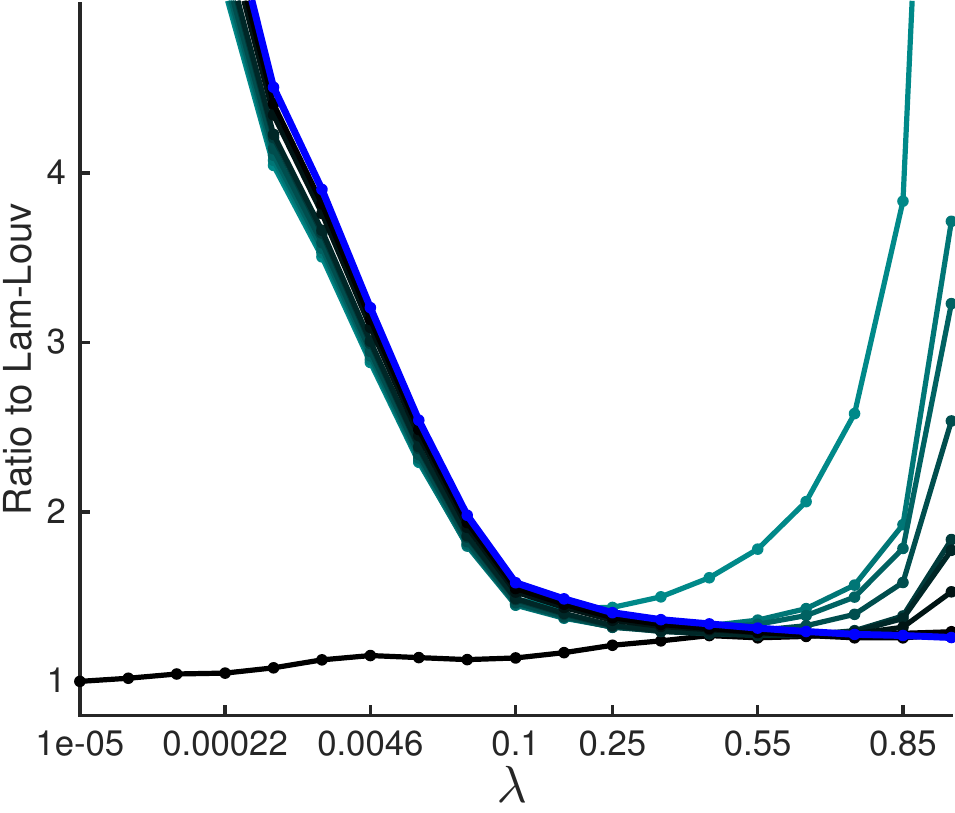}}
	\caption{As we increase the number of clusters formed by Graclus (left), the algorithm does better for large values of~$\lambda$ and worse for
		small values. The algorithm seems particularly well-suited to optimize the \lcc objective for very small values of~$\lambda$. Darker curves represent
		a larger number of clusters formed; the number of clusters formed ranges from~$2$ on the far left of the plot to just over~$2200$ for the right-most
		curve. In the right plot we vary the density~$\rho$ of quasi-cliques formed by RMQC from~$0.5$ to~$0.85$. In this plot, darker curves represent a larger density. As density increases, the curves converge to the performance of RMC, shown in blue.}
	\label{many}
\end{figure}

\subsection{Cliques in Large Collaboration Networks}
The connection between \lcc and cluster deletion provides a new approach for enumerating groups in large networks.
Here we evaluate \alg{GrowClique} for cluster deletion and use it to cluster two large collaboration networks,
one formed from a snapshot of the author-paper DBLP dataset in~2007, and the other generated using actor-movie information from the NotreDame actors
dataset~\cite{Barabasi}. The original data in both cases is a bipartite network indicating which \emph{players} (i.e., authors or actors) have parts in different \emph{projects} (papers or movies respectively). We transform each bipartite network into a graph in which nodes are players and edges represent collaboration on a project. 

At each iteration \alg{GrowClique} grows~$500$ (possibly overlapping) cliques from random seeds and selects the largest to be included in
the final output. We compare against \alg{RMC}, an expensive method which provably returns a $2$-approximation to the optimal cluster deletion objective~\cite{dessmark2007edgeclique}. We also design \alg{ProjectClique}, a method that looks at the original bipartite network and recursively identifies the project associated with the largest number of players not yet assigned to a cluster. These players form a clique in the collaboration network, so \alg{ProjectClique} clusters them together, then repeats the procedure on remaining nodes.


Table~\ref{tab:collab} shows that \alg{GrowClique} outperforms \alg{ProjectClique} in both cases, and slightly outperforms \alg{RMC} on the actor network. Our method is therefore competitive against two algorithms that in some sense have an unfair advantage over it: \alg{ProjectClique} employs knowledge not available to \alg{GrowClique} regarding the original bipartite dataset, and \alg{RMC} performs very well mainly because it solves an NP-hard problem at each step.
\begin{table}[t]
	\caption{Cluster deletion scores for \alg{GrowClique} (GC), \alg{ProjectClique} (PC) and \alg{RMC} on two collaboration networks. \alg{GrowClique} is agnostic to the underlying player-project bipartite network, and does not solve an NP-hard objective at each iteration, yet returns very good results. Best score for each
		dataset is \emph{emphasized}.}
	\label{tab:collab}
	\centering
	\begin{tabular}{lccccc}
		\toprule
		{Dataset} & Nodes&Edges & \alg{GC} & \alg{PC} & \alg{RMC} \\
		\midrule
			Actors & 341,185&10,643,420 & \emph{8,085,286} & 8,086,715 & 8,087,241\\
			DBLP &  526,303&1,616,814 & 945,489 & 946,295 & \emph{944,087}  \\
		\bottomrule
	\end{tabular}
\end{table}
\subsection{Clustering Yeast Genes} 
The study of cluster deletion and cluster editing (which is equivalent to $\pm 1$-correlation clustering) was originally motivated by applications to clustering genes using expression patterns~\cite{bendor1999clustering,shamir2004cluster}. Standard \alg{LambdaCC} is a natural framework for this, since it generalizes both objectives and interpolates between them as~$\lambda$ ranges from~$1/2$ to~$m/(m+1)$.
%
%
We cluster genes of the \emph{Saccharomyces cerevisiae} yeast organism using microarray expression data collected by Kemmeren
et al.~\cite{kemmeren2014large}. With the~$200$ expression values from the dataset, we compute correlation coefficients between all pairs of genes.
We threshold these at~$0.9$ to obtain a small graph of~$131$ nodes corresponding to unique genes, which we cluster with \alg{twoCD}.
For this cluster deletion experiment, our algorithm returns the optimal solution: solving the LP-relaxation returns a solution that is in fact integral. We validate each clique of size at least three returned by \alg{twoCD} against known gene-association data from the Saccharomyces Genome Database (SGD) and the String Consortium Database (see Table~\ref{tab:genes}). With one exception, these cliques match groups of genes that are known to be strongly associated, according to at least one validation database.
The exception is a cluster with four genes (YHR093W, YIL171W, YDR490C, and YOR225W), three of which, according to the SGD are not known to be associated with any Gene Ontology term. 
We conjecture that this may indicate a relationship between genes not previously known to be related. 
\begin{table}[h!]
	\caption{We list cliques of size~$\geq 3$ in the optimal clustering (found by \alg{twoCD}) of a network of~$131$ yeast genes. We
		validate each cluster using the SGD GO slim mapper tool, which identifies any GO term (function, process, or component of the organism) for a given
		gene. We list one GO term shared by all genes in the cluster, if one exists. The Term~$\%$ column reports the percentage of all genes in the organism
		associated with this term. A low percentage indicates a cluster of genes that share a process, component, or function that is not widely shared among other genes.
		The final column shows the minimum String association score between every pair of genes in the cluster, a number between~$0$ and~$1000$ (higher is
		better). Any non-zero score is a strong indication of gene association, as the majority of String scores between genes of \emph{S.~cerevisiae} are
		zero. All clusters, except the third, either have a high minimum String score or are all associated with a specific GO term.}
	\label{tab:genes}
	\centering
	\begin{tabular}{lllll}
		\toprule
		Clique \#& Size & Shared GO term & Term $\%$ & String \\
		\midrule
		1&6 & nucleus & 34.3 & 0\\
		2&4 & nucleus & 34.3 & 202\\
		3&4 & N/A	& - & 0\\
		4&4 & vitamin metabolic process & 0.7 & 980\\
		5&3 & cytoplasm & 67.0 &990 \\
		6&3 & cytoplasm & 67.0 & 998 \\
		7&3 & N/A & - & 962 \\
		8&3 & cytoplasm & 67.0 & 996 \\
		9&	3 & N/A & - & 973 \\
		10&	3 & transposition & 1.7 & 0\\
		\bottomrule
	\end{tabular}
\end{table}
\subsection{Social Network Analysis with \lcc}
Clustering a social network using a range of resolution parameters can reveal valuable insights about how links are formed in the network. Here we examine several graphs from the Facebook100 dataset, each of which represents the induced subgraph of the Facebook network corresponding to a US university at some point in~2005. The networks come with anonymized meta-data, reporting attributes such as major and graduation year for each node.
While meta-data attributes are not expected to correspond to ground-truth communities in the network~\cite{Peel2017ground}, we
do expect them to play a role in how friendship links and communities are formed. In this experiment
we illustrate strong correlations between the link structure of the networks and the dorm, graduation year, and student/faculty
status meta-data attributes. We also see how these correlations are revealed, to different degrees, depending on our choice of~$\lambda$. 

Given a Facebook subgraph with~$n$ nodes, we cluster it with degree-weighted
\alg{Lambda-Louvain} for a range of~$\lambda$ values between $0.005/n$ and $0.25/n$. In this clustering, we refer to two nodes in the same cluster as an \emph{interior pair}. We measure how well a meta-data attribute~$M$ correlates
with the clustering by calculating the proportion of interior pairs that share the same value for~$M$. This value, denoted by~$P(M)$, can also be interpreted as the probability of selecting an interior pair uniformly at random and finding that they agree on attribute~$M$. To determine whether the probability is meaningful, we compare it against a null probability $P(\tilde{M})$: the probability that a random interior pair agree at a \emph{fake} meta-data
attribute~$\tilde{M}$. We assign to each node a value for the fake attribute~$\tilde{M}$ by performing a random
permutation on the vector storing values for true attribute~$M$. In this way, we can compare each true attribute~$M$
against a fake attribute~$\tilde{M}$ that has the same exact proportion of nodes with each attribute value, but does not impart any true information regarding each node. 

We plot results for each of the three attributes $M \in \{\dorm, \, \yr, \, \fs \text{ (student/faculty)}\}$ on four
Facebook networks in Figure~\ref{fig:facebook}, as~$\lambda$ is varied. In all cases, we see significant differences between~$P(M)$
and~$P({\tilde{M}})$. In general,~$P(\yr)$ and $P(\fs)$ reach a peak at small values of~$\lambda$ when clusters are
large, whereas $P(\dorm)$ is highest when~$\lambda$ is large and clusters are small. This indicates that the first two attributes are more highly correlated with large sparse communities in the network, whereas sharing a dorm is more correlated with smaller, denser communities. Caltech, a small residential university, is an exception to these trends and exhibits a much stronger correlation with the dorm attribute, even for very small~$\lambda$.

\begin{figure}[t]
	\centering
	\subfloat[Swarthmore $n = 1659$ \label{swarthmore}]
	{\includegraphics[width=.35\linewidth]{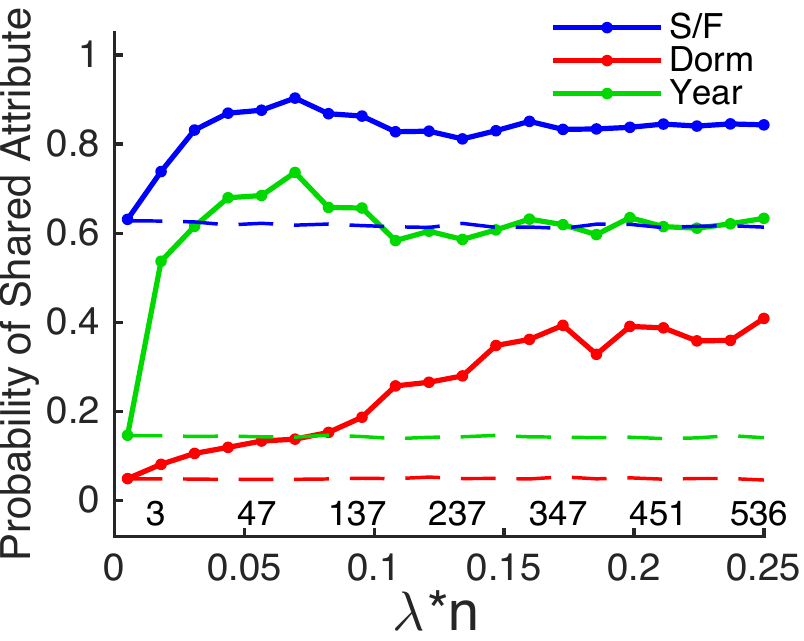}}\hspace{2cm}
	\subfloat[Yale $n = 8578$ \label{yale}]
	{\includegraphics[width=.35\linewidth]{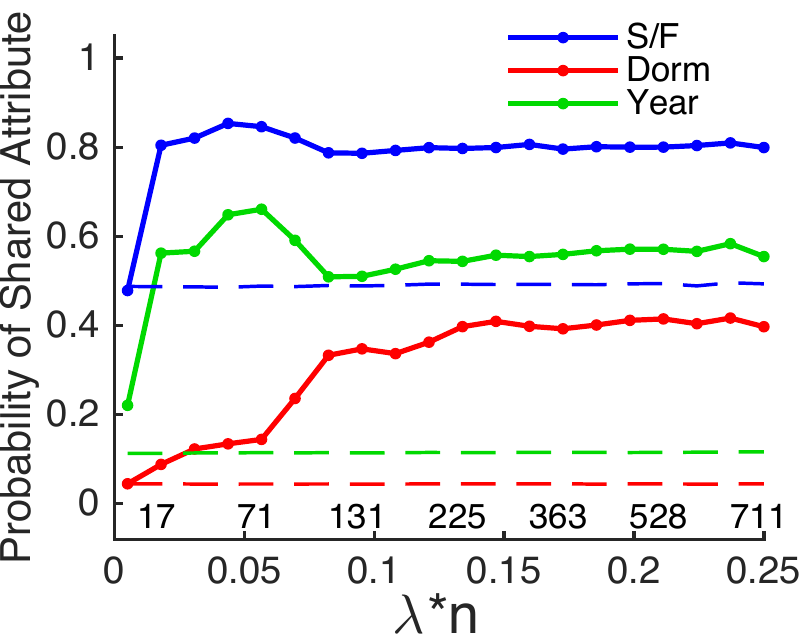}}\hspace{2cm}
	\centering
	\subfloat[Cornell\label{cornell} $n = 18660$]
	{\includegraphics[width=.35\linewidth]{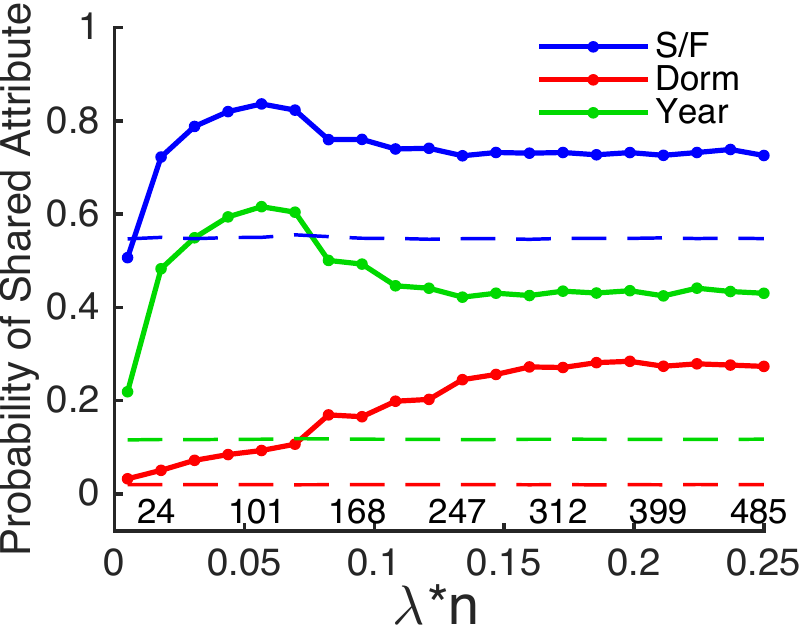}}\hspace{2cm}
	\subfloat[Caltech $n = 769$\label{caltech}]
	{\includegraphics[width=.35\linewidth]{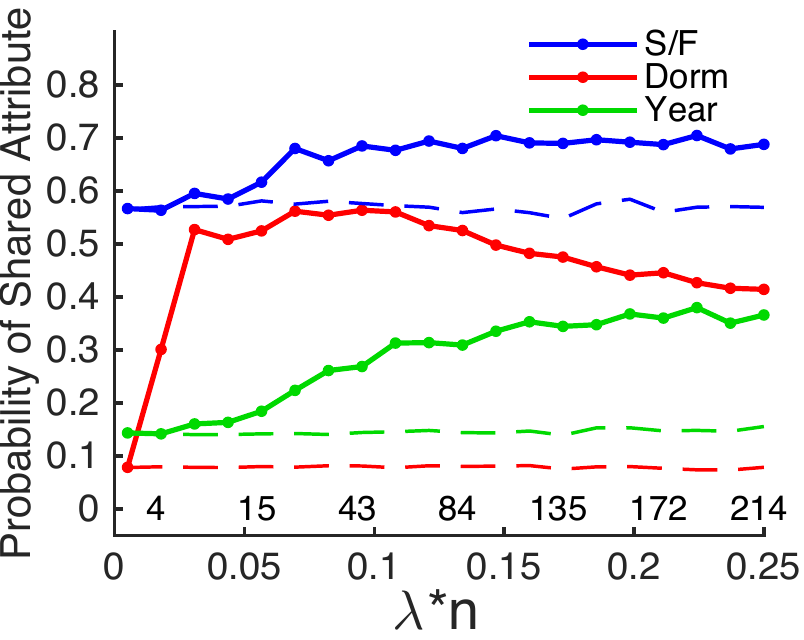}}
	\caption{On four university Facebook graphs, we illustrate that the dorm (red), graduation year (green), and student/faculty (S/F) status (blue) meta-data attributes all correlate highly with the clustering found by \alg{Lambda-Louvain} for each~$\lambda$. Above the $x$-axis we show the number of clusters formed, which strictly increases with~$\lambda$. The $y$-axis reports the probability that two nodes sharing a cluster also share an attribute value. Each attribute curve is compared against a null probability, shown as a dashed line of the same color. 
		The large gaps between each attribute curve and its null probability indicate that the link structure of all
		networks is highly correlated with these attributes. In general, probabilities for $\yr$ and $\fs$
		status are highest for small~$\lambda$, whereas $\dorm$ has a higher correlation with smaller, denser communities in the
		network. Caltech is an exception to the general trend; see the main text for discussion.}
	\label{fig:facebook}
\end{figure}

\subsection{Clustering an Email Network}
Our previous experiment highlighted our method's ability to detect strong correlations between meta-data attributes and community structure in real-world networks. We now explore how to appropriately choose a resolution parameter that best highlights the extent to which meta-data relates to good clusters in a graph. In particular, we cluster the email-EU graph, which encodes email correspondence between~1005 faculty members organized into~42 departments (the meta-data) at a European research institution~\cite{leskovec2007graph}. In order to learn as much as we can about the relationship
between our resolution parameter and the meta-data, we purposely select the value of~$\lambda$ that empirically leads
to the best Adjusted Rand Index (ARI) scores between the clustering determined by departments and the output of degree-weighted \alg{Lambda-Louvain}. We find that when $\lambda = 10^{-4}$, our method's ARI score is much higher than the scores obtained by algorithms that optimize a more rigid objective function (Table~\ref{tab:eu}). This highlights the potential benefit our framework can provide when given the right parameter, and shows the importance of developing good techniques for appropriately selecting~$\lambda$.

The insight in this experiment comes from comparing normalized cut scores in different clusterings. Running \alg{Lambda-Louvain} with $\lambda = 10^{-4}$ yields clusters with scaled normalized cut between $0.8825\times10^{-5}$ and $3.129\times 10^{-5}$, similar to the clusters defined by meta-data,
which exhibit scores between~$1.026\times 10^{-5}$ and $3.126 \times 10^{-5}$. Interestingly, these values are roughly an order-of-magnitude smaller than our choice of~$\lambda = 10^{-4}$. This is consistent with our result in Theorem~\ref{thm:sscbound}:~$\lambda$ is an upper bound on the scaled normalized cut
scores of \emph{all} clusters formed by \alg{Lambda-Louvain}. This suggests a general strategy for setting~$\lambda$ when we wish to learn the extent to which meta-data and community structure coincide. If we are given any a priori knowledge about the scaled normalized cut score for node sets sharing the same attribute, we know not to set~$\lambda$ equal to or lower than this score. Rather, we choose a resolution parameter that is not too far from this value, but is still a \emph{generous} upper bound. In future work, we aim to continue researching both theoretically and experimentally how to more precisely determine a priori the right upper-bound~$\lambda$ to use.

\begin{table}[t]
	\caption{To cluster the email-EU network, we run each method~20 times and report the median ARI score
		between each clustering and the clustering determined by meta-data. When $\lambda = 10^{-4}$, degree-weighted \alg{Lambda-Louvain} exhibits the highest ARI scores, indicating that our resolution parameter gives us the flexibility to better detect the extent to which meta-data coincides with community structure. 
		We have Metis~\cite{karypis1998fast} form~27 clusters and Graclus form~13, since each yields the best results for the algorithm.} 
	\label{tab:gtruth}
	\centering
	\begin{tabular}{ccccc}
		\toprule
		Lam-Louv & {Metis} & {Graclus} & {Louvain} & {InfoMap} \\
		\midrule
		0.587& 0.359 & 0.393 & 0.264 & 0.273\\
		\bottomrule
	\end{tabular}
	\label{tab:eu}
\end{table}


\section{Discussion}
We have introduced a new clustering framework that unifies several other commonly-used objectives and offers many attractive theoretical
properties. We prove that our objective function interpolates between the sparsest cut objective and the cluster deletion problem, as we vary a single input parameter,~$\lambda$. We give a $3$-approximation algorithm for our objective when $\lambda \geq 1/2$, and a related method which improves the best approximation factor for cluster deletion from 3 to 2. We also give scalable procedures for greedily improving our objective, which are successful in a wide variety of clustering applications. These methods are easily modified to add must-cluster and cannot-cluster constraints, which makes them amenable to many applications. In future work, we will continue exploring approximations when~$\lambda< 1/2$. 

\section*{Acknowledgements} This work was supported by several funding agencies: Nate Veldt and David Gleich are supported by NSF award IIS-1546488, David Gleich is additionally supported by NSF awards CCF-1149756 and CCF-093937 as well as the DARPA Simplex program and the Sloan Foundation. Anthony Wirth is supported by the Australian Research Council. We thank Flavio Chierichetti for several helpful conversations and also thank the anonymous reviewers for several helpful suggestions for improving our work, in particular for mentioning connections to the work of van Zuylen and Williamson~\cite{zuylen2009deterministic}, which led to significantly improved approximation results.

\newpage
\bibliographystyle{abbrv}
\bibliography{all-bibliography} 

\end{document}